\documentclass[prd,aps,amsfonts,superscriptaddress,nofootinbib,longbibliography,notitlepage,eqsecnum]{revtex4-1}

\usepackage{graphicx}
\usepackage{xcolor,verbatim}
\usepackage{caption}
\usepackage{rotating}
\usepackage{amsmath,amssymb,graphics,amsthm,isomath}
\usepackage{subfigure}
\usepackage{physics }
\usepackage{circuitikz}
\usepackage[colorlinks=true, urlcolor=violet, linkcolor=blue, citecolor=red, hyperindex=true, linktocpage=true]{hyperref}
\usepackage[capitalise,compress]{cleveref}

\newtheorem{thm}{Theorem}
\newtheorem{cor}[thm]{Corollary}

\newtheorem{obs}[thm]{Observation}

\theoremstyle{definition}
\newtheorem{defn}[thm]{Definition}

\makeatletter
\renewcommand{\p@subsection}{}
\renewcommand{\p@subsubsection}{}
\makeatother

\usepackage{xcolor}
\usepackage{mathtools}

\makeatletter
\pgfcircdeclarebipolescaled{instruments}
{
    \anchor{text}{\pgfextracty{\pgf@circ@res@up}{\northeast}
        \pgfpoint{-.5\wd\pgfnodeparttextbox}{
            \dimexpr.5\dp\pgfnodeparttextbox+.5\ht\pgfnodeparttextbox+\pgf@circ@res@up\relax
        }
    }
}
{\ctikzvalof{bipoles/oscope/height}}
{josephson}
{\ctikzvalof{bipoles/oscope/height}}
{\ctikzvalof{bipoles/oscope/width}}
{
    \pgf@circ@setlinewidth{bipoles}{\pgfstartlinewidth}
    \pgfextracty{\pgf@circ@res@up}{\northeast}
    \pgfextractx{\pgf@circ@res@right}{\northeast}
    \pgfextractx{\pgf@circ@res@left}{\southwest}
    \pgfextracty{\pgf@circ@res@down}{\southwest}
    \pgfmathsetlength{\pgf@circ@res@step}{0.25*\pgf@circ@res@up}
    \pgfscope
        \pgfpathrectanglecorners{\pgfpoint{\pgf@circ@res@left}{\pgf@circ@res@down}}{\pgfpoint{\pgf@circ@res@right}{\pgf@circ@res@up}}
        \pgf@circ@draworfill
    \endpgfscope
    \pgfscope
      \pgfpathmoveto{\pgfpoint{\pgf@circ@res@left}{\pgf@circ@res@up}}%
      \pgfpathlineto{\pgfpoint{\pgf@circ@res@right}{\pgf@circ@res@down}}%
      \pgfpathmoveto{\pgfpoint{\pgf@circ@res@right}{\pgf@circ@res@up}}%
      \pgfpathlineto{\pgfpoint{\pgf@circ@res@left}{\pgf@circ@res@down}}%
      \pgfusepath{draw}
    \endpgfscope
}
\def\pgf@circ@josephson@path#1{\pgf@circ@bipole@path{josephson}{#1}}
\tikzset{josephson/.style = {\circuitikzbasekey, /tikz/to path=\pgf@circ@josephson@path, l=#1}}

\usepackage{dsfont}

\begin{document}

\title{Symplectic geometry and circuit quantization}

\author{Andrew Osborne}
\email{andrew.osborne-1@colorado.edu}
\affiliation{Department of Physics, University of Colorado Boulder, Boulder CO 80309, USA}
\affiliation{Center for Theory of Quantum Matter, University of Colorado, Boulder CO 80309, USA}

\author{Trevyn Larson}
\affiliation{Department of Physics, University of Colorado Boulder, Boulder CO 80309, USA}
\affiliation{National Institute of Standards and Technology, Boulder, Colorado 80305, USA}

\author{Sarah Jones}
\affiliation{Department of Physics, University of Colorado Boulder, Boulder CO 80309, USA}
\affiliation{Department of Electrical, Computer \& Energy Engineering, University of Colorado Boulder, Boulder, CO 80309, USA}

\author{Ray W. Simmonds}
\affiliation{National Institute of Standards and Technology, Boulder, Colorado 80305, USA}

\author{Andr\'as Gyenis}
\affiliation{Department of Electrical, Computer \& Energy Engineering, University of Colorado Boulder, Boulder, CO 80309, USA}

\author{Andrew Lucas}
\email{andrew.j.lucas@colorado.edu}
\affiliation{Department of Physics, University of Colorado Boulder, Boulder CO 80309, USA}
\affiliation{Center for Theory of Quantum Matter, University of Colorado, Boulder CO 80309, USA}

\begin{abstract}
   Circuit quantization is an extraordinarily successful theory that describes the behavior of quantum circuits with high precision. The most widely used approach of circuit quantization relies on introducing a classical Lagrangian whose degrees of freedom are either magnetic fluxes or electric charges in the circuit.  By combining nonlinear circuit elements (such as Josephson junctions or quantum phase slips), it is possible to build circuits where a standard Lagrangian description (and thus the standard quantization method) does not exist.  Inspired by the mathematics of symplectic geometry and graph theory, we address this challenge, and present a Hamiltonian formulation of non-dissipative electrodynamic circuits.  The resulting procedure for circuit quantization is  independent of whether circuit elements are linear or nonlinear, or if the circuit is driven by external biases.  We explain how to re-derive known results from our formalism, and provide an efficient algorithm for quantizing circuits, including those that cannot be quantized using existing methods. 
\end{abstract}

\date{\today}
\maketitle

\tableofcontents

\section{Introduction}

The quantum mechanical description of superconducting circuits has paved the way for the rapid evolution of superconductor-based quantum computers~\cite{doi:10.1126/science.1231930, krantz_quantum_2019, kjaergaard_superconducting_2020,doi:10.1063/5.0082975}, enabled the discovery of the transmon~\cite{koch_2017}, fluxonium~\cite{manucharyan_fluxonium_2009}, bosonic~\cite{grimsmo_quantum_2021} and more complex circuits~\cite{gyenis_2021}, facilitated the advancements of coupling between qubits~\cite{stehlik_2021}, opened new avenues towards quantum simulation~\cite{PRXQuantum.2.017003}, and led to the theory of circuit quantum electrodynamics~\cite{blais_circuit_2021}. In this established formalism of circuit quantization~\cite{yurke_quantum_1984, devoret_leshouches, vool_introduction_2017, nigg_black-box_2012, ulrich_dual_2016, riwar_circuit_2022, chitta_computer-aided_2022,rajabzadeh2022analysis, koliofoti_compact_2022, rymarz_consistent_2022}, generalized branch or node fluxes, or charges, describe the energy of the elements.  While these fluxes and charges are conjugates, in all but the simplest circuits, there are inevitable \emph{constraints} that arise between variables, complicating a straightforward quantization prescription.

\subsection{The standard approach to circuit quantization}

The standard resolution in the literature is to begin by studying the classical Lagrangian mechanics of the circuit. 
In the Lagrangian formalism, we efficiently remove non-dynamical degrees of freedom by integrating them out; after this step, we perform a Legendre transformation to a classical Hamiltonian for the genuinely dynamical degrees of freedom.   Since the Legendre transformation reveals the canonical momenta for each coordinate, we can quantize a Hamiltonian self-consistently.

 As a simple example, consider an inductor, $L$, with two capacitors, $C_1$ and $C_2$, all in parallel [see Fig.~\ref{fig:intro_example}(a)].  There is one degree of freedom, which can be identified as the flux $\phi$ across the inductor.  Since $\dot \phi$ is the voltage drop across the capacitor, one writes down a Lagrangian \begin{equation}
    L = \frac{1}{2}(C_1+C_2)\dot \phi^2 -\frac{1}{2L}\phi^2.
\end{equation}
 This Lagrangian is interpreted as a ``kinetic energy minus potential energy" term, and is mathematically equivalent to a simple pendulum.  Note that we are also able to elegantly handle the two capacitors in parallel -- the Lagrangian automatically adds them into a single effective capacitor for us.  With a Lagrangian at hand, we find the conjugate momentum \begin{equation}
    q = \frac{\partial L}{\partial \dot \phi} = (C_1+C_2)\dot \phi,
\end{equation}
and finally the Hamiltonian reads\begin{equation}
    H = \frac{1}{2(C_1+C_2)}q^2 +\frac{1}{2L}\phi^2,
\end{equation}
where $\phi$ and $q$ are conjugate variables, and their Poisson bracket is $\{\phi,q\}=1$. We can quantize the circuit based on these conjugate pairs by imposing canonical commutation relations
\begin{equation}
    \left[\hat\phi,\hat{q}\right] = \mathrm i\hbar, \label{eq:introbracket}
\end{equation} 
where $\hbar$ is the reduced Planck constant, $\hat\phi$ is the flux operator, and $\hat{q}$ is the charge operator.

In the example above, we have linear capacitors and inductors.  Circuits might also contain a combination of  nonlinear and noninvertible capacitive and inductive elements, for example, Josephson junctions (JJs)~\cite{likharev_JosephsonJunctions, Likharev1985} and quantum phase slips (QPSs)~\cite{mooij_superconducting_2006, Astafiev2012, Kerman2013, Shaikhaidarov2022}.   Treating both of these nonlinear and noninvertible elements in the same circuit is still an open problem: because the energy of the JJs depends on the flux $\phi$ as $E\sim \cos(2\pi\phi/\phi_0)$, while the energy of the QPSs on the charge $q$ across the element as $E\sim \cos(2\pi q/2e)$, one can prove that no Lagrangian of the circuit with a single type of variable (flux or charge) exists in general.\footnote{A Lagrangian can be defined locally, but not globally, on phase space.  The obstruction to a global Lagrangian is that the Legendre transform is not defined because $\partial^2 L/\partial \dot \phi^2$ (or $\partial^2 L/\partial \dot q^2$) is not always positive definite.} Here, $\phi_0=h/2e$ is the fundamental flux quantum, $h$ is the Planck constant, and $e$ is the electron charge. Although for the minimal circuit of one JJ and one QPS included in a loop [see Fig.~\ref{fig:intro_example}(b)], one can immediately write down a Hamiltonian~\cite{le_doubly_2019} \begin{equation}
    H = -E_Q \cos \left(2\pi\frac{ q}{2e}\right) - E_J \cos \left(2\pi\frac{\phi}{\phi_0}\right), \label{eq:dualmonintro}
\end{equation}  circuits involving even a few of these elements can involve non-trivial constraints [see Fig.~\ref{fig:intro_example}(c)]. For example, the number of degrees of freedom is not equal to the number of JJs or QPSs; worse, the fluxes and charges across the different elements may not be conjugate pairs.   Understanding how to even identify the dynamical degrees of freedom, let alone quantize the circuit, is an open problem.

\subsection{Our quantization method}

In this paper, we present an alternative approach to circuit quantization.  Our approach is inspired by earlier work \cite{burkard_multilevel_2004} that links graph theory to circuit quantization.   However, rather than using Lagrangian mechanics as the starting point, we instead appeal to the mathematics of symplectic geometry \cite{arnold,dasilva}, which generalizes textbook Hamiltonian mechanics to more abstract/general settings. We will show that this more abstract perspective elegantly resolves the puzzle of how to choose canonical coordinates, without relying on any assumptions about the constitutive relations of the inductive or capacitive elements in the circuit.  Hence, our approach is universal for quantum circuits made out of nonlinear inductive and capacitive elements and capable of resolving the puzzle above.  

The key insight of our theory is that the most natural quantization prescription involves building conjugate degrees of freedom out of charge variables on branches ($q_e$) and flux variables on nodes ($\phi_v$) of the capacitive subgraph of the circuit (see Fig.~\ref{fig:intro}).  To understand our construction, let us remind the reader how one usually models circuits.  The degrees of freedom are voltages $V$ and currents $I$, which can be integrated in time to give flux $\phi$ and charge $q$.  In circuits, $\phi$ and $q$ are canonically conjugate variables, similar to position and momentum: recall Eq.~(\ref{eq:introbracket}). Of course, a typical circuit involves multiple elements and therefore multiple flux and charge variables.  Due to Kirchhoff's voltage law (the sum of the voltage drops around a loop vanishes in the absence of external magnetic fields), it is natural to define voltages on nodes (or vertices) $v$ of the circuit. On the other hand, currents $I$ are naturally defined on branches (or edges) $e$ in the circuit.  In general, there is not a one-to-one mapping between the branches $v$ and nodes $e$ of a circuit. How, therefore, can we possibly find the conjugate pairs? 

\begin{figure}
    \centering
    \includegraphics[width = 8.6cm]{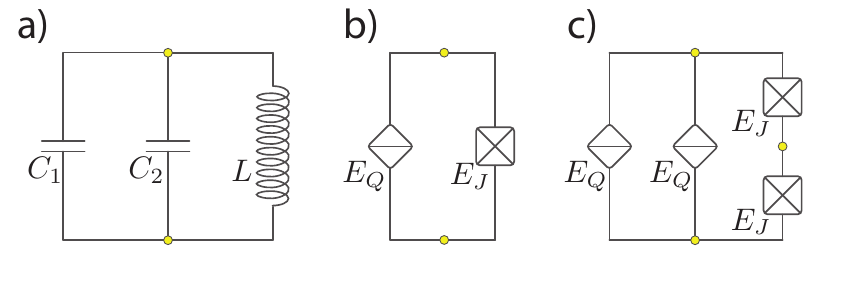}
    \caption{\textbf{Examples of quantum circuits with Lagrangians using one or two types of generalized coordinates.} (a) An inductor shunted by two capacitors: the Lagrangian of the circuit contains a single type of variable, the flux across the inductor, such that $L(\phi,\dot{\phi})$. The charge-flux conjugate pairs are defined at the Lagrangian level. (b) A Josephson junction and a quantum phase slip junction forming a loop. This quantum circuit can not be described by a Lagrangian using a single type of variable but only with a Lagrangian that contains both charge and flux variables, $L(\phi,\dot{\phi}, q,\dot{q})$. (c) A more complex circuit of multiple Josephson junctions and quantum phase slips, where writing down a Hamiltonian requires geometrical arguments.} 
    \label{fig:intro_example}
\end{figure}

The mathematical puzzle above is not entirely semantic.   In classical and quantum Hamiltonian mechanics, there must be an equal number of position and momentum coordinates.  As we described above, the prior resolution in the literature has been to use Lagrangian mechanics to avoid tackling this  issue head-on.   The Lagrangian is written using only one type of variable (for example, branch fluxes), and then the conjugate momenta is found at the Lagrangian level.  Any excess in degrees of freedom is dealt with by integrating out non-dynamical variables. However, writing down a Lagrangian as a first step is not always possible.  A simple example is the circuit that we discussed above, the dualmon qubit~\cite{le_doubly_2019} shown  in Fig.~\ref{fig:intro_example}(b). The Hamiltonian in Eq.~(\ref{eq:dualmonintro}) is only known because there is just one dynamical degree of freedom.

What this paper provides is a way of solving the constraints on $\phi$ and $q$ variables, directly in a Hamiltonian formulation, such that we can find suitable linear combinations of charge and flux variables that are canonically conjugate (and equal in number).  As stated above, the approach is inspired by symplectic geometry, together with the simple observation that the equations of motion for a  circuit follow from the action \begin{equation}
    S = \int \mathrm{d}t\left[ -E_{\mathrm{tot}}(q_e,\phi_v) + \sum_{e,v} q_e \Omega_{ev}\dot\phi_v\right]. \label{eq:mainintro}
\end{equation}
Note that this action involves both branch charges $q_e$ and node fluxes $\phi_v$.   
Remarkably, this action encodes within it the \emph{Hamiltonian mechanics} of the circuit.  The function $E_{\mathrm{tot}}$ (upon fixing all constrained variables, which we provide a generic prescription to do) is the Hamiltonian itself, and equal to the sum of inductive and capacitive energies in the circuit. Critically, the energy of elements are expressed in their native coordinates: inductive elements' energies are expressed in terms of fluxes $\phi_v$, while capactive elements' energies are expressed in terms of charges $q_e$. Furthermore,   $\Omega_{ev}$ is the \emph{incidence matrix} for the \emph{capacitive subgraph} of the circuit (see Fig.~\ref{fig:intro}).   $q_e \Omega_{ev}\dot\phi_v$  is naturally interpreted using  symplectic geometry, and implies both a classical Poisson bracket, and a quantization prescription.

The rest of the paper will derive Eq.~(\ref{eq:mainintro}) in Section \ref{sec:formal}, explain how to subsequently quantize circuits in Section \ref{sec:qm}, and then show how to identify the physical degrees of freedom in numerous example circuits in Section \ref{sec:examples}.  We have written this paper in a pedagogical and self-contained way; no prior knowledge is required in either circuit quantization or mathematical physics (beyond textbook Lagrangian and Hamiltonian mechanics).  

\begin{figure}
    \centering
    \includegraphics[width = 17.2cm]{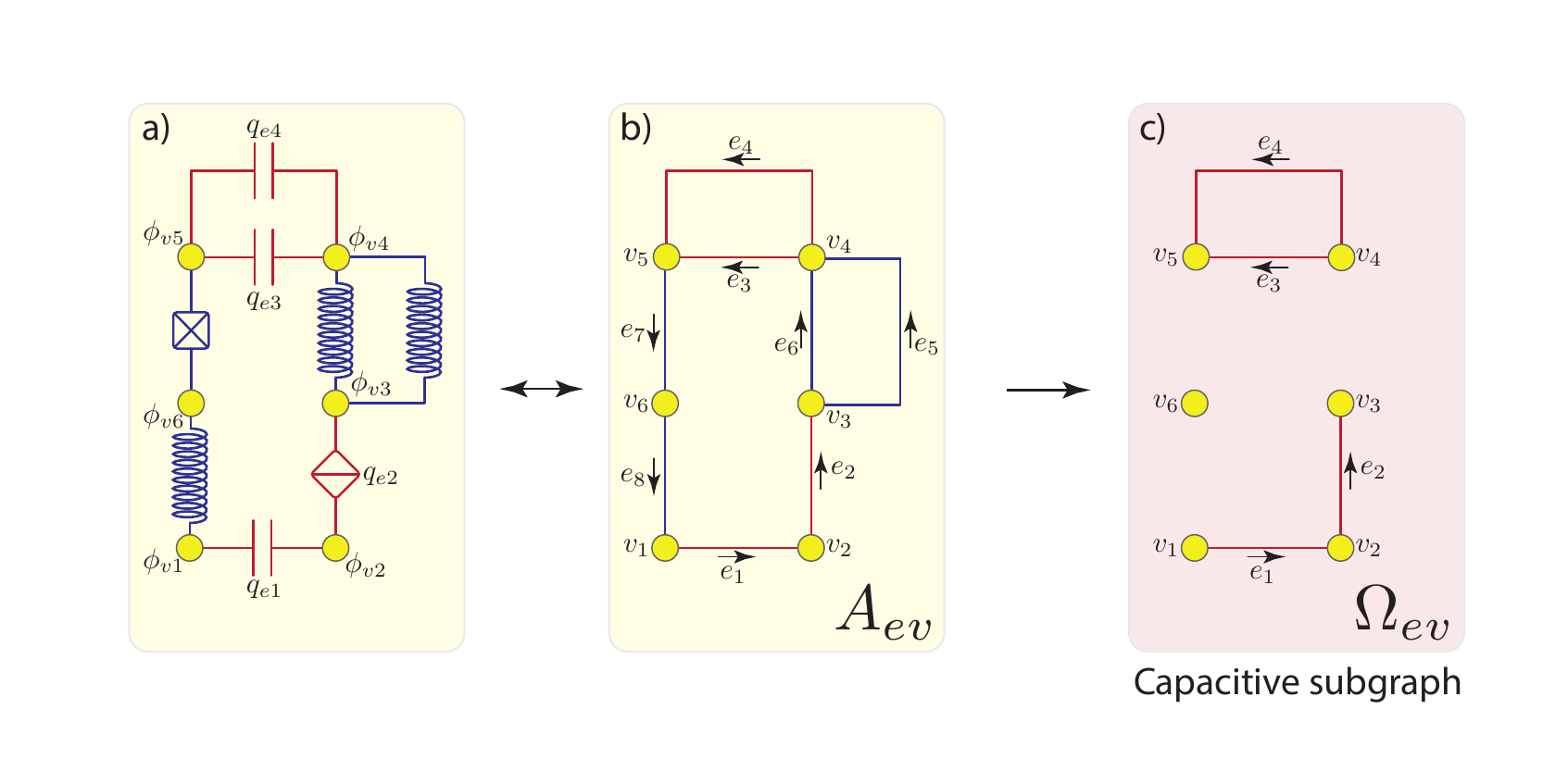}
    \caption{\textbf{Graph theory and quantum circuits.} (a) An arbitrary circuit containing all four types of superconducting circuit elements: capacitors, inductors, Josephson junctions and quantum phase slips. The node fluxes are $\phi_{v_i}$, while the branch charges across the capacitive elements are $q_{e_i}$. (b) The graph of the full circuit. The vertices of the graph are label as $v_i$, while the edges are denoted as $e_i$. Inductive branches are colored with blue lines, while capacitive ones are highlighted with red lines. The incidence matrix of the graph of the full circuit is $A_{ev}$. (c) The capacitive subgraph of the circuit containing only capacitive edges. The capacitive incidence matrix is $\Omega_{ev}$.} 
    \label{fig:intro}
\end{figure}

\section{Classical formalism}\label{sec:formal}
We now begin a gentle introduction to the classical mechanics of quantum circuits.  Our focus will be to motivate the derivation of Eq.~(\ref{eq:mainintro}); however, we provide background knowledge both into the experimental systems and also the mathematics of Hamiltonian mechanics, as is necessary to appreciate Eq.~(\ref{eq:mainintro}).

\subsection{Circuit elements and degrees of freedom}
First, we review the definition of branch charges and node fluxes. When describing superconducting circuits, we assume that the circuit elements are connected by perfect superconducting wires without inductive or resistive contributions. We call a part of the circuit that contains a circuit element a branch, and an intersection of the superconducting wires a node. We will soon mathematically describe the circuits as graphs, where branches will be associated to edges in the graph, and nodes as  vertices of the graph. We will use the former terminology throughout the paper to conform with the tradition used in quantum circuits.  As an homage to the mathematics, however, we will use the letter $e$ to denote a generic branch (edge), and $v$ to denote a generic node (vertex).

Now, consider a two-terminal superconducting circuit element defined between two nodes of a circuit. The node voltage $V_v(t)$ is the voltage at a given node, and the branch current $I_e(t)$ is the current flowing through the circuit element. We assume that it is sufficient to use only the voltage at nodes and the current across the element to describe the system and ignore the current and voltage distribution inside the elements, i.~e., we use a lumped element approximation for the circuit.\footnote{We will always separate out the inductive and capacitive elements, rather than working with circuit elements whose impedances are complicated functions of frequency. In principle, one sometimes does not wish to do this, but lumping inductors and capacitors together will complicate the quantization prescription in this paper.} Next, we define the generalized node flux $\phi_v$ and the branch charges $q_e$ as the time integral of the voltage and the current
\begin{subequations}
\begin{align}
\phi_v(t) &= \int\limits_{-\infty}^t \mathrm{d}\tau V_v(\tau), \\
q_e(t) &= \int\limits_{-\infty}^t \mathrm{d}\tau I_e(\tau).
\end{align}
\end{subequations}
It is also common to define the generalized branch fluxes, as the difference between the fluxes of the corresponding nodes; if branch $e$ connects nodes $v_i$ and $v_j$, the branch flux is 
\begin{equation}\label{eqn:branchflux}
\phi_e(t) = \phi_{v_i}(t) - \phi_{v_j}(t).
\end{equation}

The most standard variables in which one performs circuit quantization are $\phi_v$ or $\phi_e$, starting in the Lagrangian formalism.  However, such coordinates cannot describe circuits with quantum phase slip elements due to their non-invertible charge-voltage relationship. 
 To address that challenge, loop charges have been also used as an alternative approach to describe charge degrees of freedom in the Lagrangian description, in the special case where the graph of the circuit is planar~\cite{ulrich_dual_2016}.

The various circuit elements establish different relations between the branch variables and their derivatives [see Fig.~\ref{fig:elements}]. For example, linear capacitors and inductors connect linearly two variables
\begin{subequations}
\begin{align}
\dot q_e &= \frac{1}{L} \phi_e, \\
\dot \phi_e &= \frac{1}{C}  q_e, 
\end{align}
\end{subequations}
where $C$ is the capacitance, $L$ is the inductance, and for simplicity, we dropped the explicit time-dependence of the variables in the notations.

\begin{figure}
    \centering
    \includegraphics[width = 8.6cm]{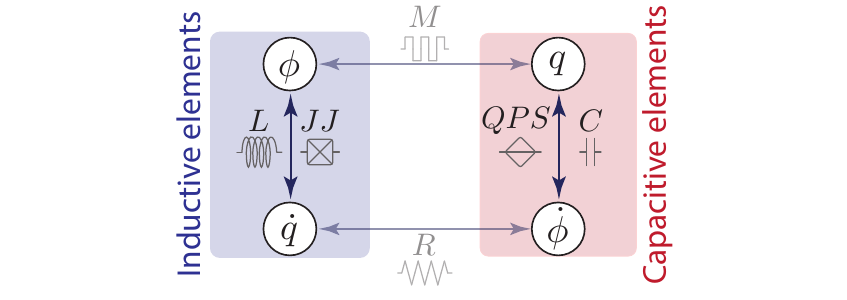}
    \caption{\textbf{Circuit elements and their constitutive relations.} The circuit elements that we consider in this work connect two variables (flux, charge, and their derivatives). Inductive elements relate flux $\phi$ and the current $\dot{q}$, such as linear inductors ($L$) and Josephson junctions ($JJ$). On the other hand, capacitive elements connect charge $q$ with voltage $\dot\phi$, such as linear capacitors ($C$) and quantum phase slip elements ($QPS$). The memristor ($M$) and resistor ($R$) connect directly charge with flux or their derivatives.} 
    \label{fig:elements}
\end{figure}

The two most well-known nonlinear and nondissipative circuit elements are the (multi-channel, low-transmission) Josephson junctions and the quantum phase slip junctions, where the connection between the variables is sinusoidal instead of linear
\begin{subequations}
\begin{align}
\dot q_e &= I_C\sin\left(2\pi\frac{\phi_e}{\phi_0}\right), \\
\dot \phi_e &= V_Q\sin\left(2\pi\frac{ q_e}{2e}\right). 
\end{align}\end{subequations}
Here, $I_C$ is the critical current of the Josephson junction, and $V_Q$ is the voltage amplitude of the quantum phase slip junction. There are other more complicated relations, e.g. for high-transmission Josephson junctions, containing higher harmonics in the current-phase relationship~\cite{PhysRevLett.67.3836}.

In this work, we focus only on non-dissipative circuits; thus, the two types of elements that we are concerned with are the capacitive and inductive elements. The energy stored in the elements can be expressed as \begin{equation}
    E=\int\limits_{-\infty}^t\mathrm{d}\tau I_e(\tau) V_e(\tau) =  \int\limits_{-\infty}^t\mathrm{d}\tau\dot q_e(\tau)\dot\phi_e(\tau).
\end{equation} 
We then see that the energy of the capacitive elements $E_\mathcal{C}$ depends only on the branch charge $q_e$ such that $E_\mathcal{C}= E_\mathcal{C}(q_e)$, while the energy of the inductive elements $E_\mathcal{I}$ is a function of only the branch flux $\phi_e$ and $E_\mathcal{I}=E_\mathcal{I}(\phi_e)$. For example, 
\begin{subequations}
\begin{align}
\mathrm{linear\ capacitor:\ } E_\mathcal{C}(q_e) &= \frac{q_e^2}{2C}, \\
\mathrm{quantum\ phase\ slip:\ } E_\mathcal{C}(q_e) &= -E_Q\cos\left(2\pi\frac{ q_e}{2e}\right), \\
\mathrm{linear\ inductor:\ } E_\mathcal{I}(\phi_e) &= \frac{\phi_e^2}{2L}, \\
\mathrm{Josephson\ junction:\ } E_\mathcal{I}(\phi_e) &= -E_J\cos\left(2\pi\frac{\phi_e}{\phi_0}\right).
\end{align}
\end{subequations}
where $E_J=\phi_0I_C/(2\pi)$ and $E_Q= 2eV_Q/(2\pi)$ are the Josephson and quantum phase slip energies.
The derivative of the energy with respect to the coordinates gives the voltage for capacitive elements, and the current for inductive elements
\begin{subequations} \label{eq:constitutive}
\begin{align}
V_e(t) &= \frac{\partial E_\mathcal{C}(q_e)}{\partial q_e}, \label{eq:capvoltage} \\
I_e(t) &= \frac{\partial E_\mathcal{I}(\phi_e)}{\partial \phi_e}. \label{eq:indcurrent}
\end{align}
\end{subequations}

\subsection{Circuits and graph theory}\label{sec:graph}
Electrical network graph theory~\cite{peikari1974} has played an important role in the existing theory of circuit quantization \cite{burkard_multilevel_2004,devoret_leshouches,vool_introduction_2017, PhysRevA.102.032208}. Following this approach, a quantum circuit can be considered as a directed graph, where the two-node circuit elements correspond to the edges of the graph, and the vertices of the graph are the points of the circuit where the elements connect. We consider an arbitrary circuit (see Fig.~\ref{fig:intro}) that has $k$ nodes and $K$ branches, and we denote the set of nodes as $\mathcal{V}$ and the set of branches as $\mathcal{E}$.  We assume that we do not lump together inductive and capacitive elements, so that we can classify each branch as one or the other.   Suppose there are $N$ capacitive branches and $K-N$ inductive branches.   The set of capacitive branches is $\mathcal{C}\subset \mathcal{E}$, and the set of inductive branches is $\mathcal{I} \subset \mathcal{E}$; thus, $|\mathcal{E}|=K$, $|\mathcal{C}|=N$, and $|\mathcal{I}|=K-N$. 

A key object of the graph is the incidence matrix $A_{ev}$ that provides information on the interconnection of the circuit elements. In particular, the rows/columns of the matrix correspond to the edges/vertices, and the value of a matrix element is $+1$ ($-1$) if an edge points toward (from) a vertex, otherwise it is 0:
\begin{equation}
A_{ev}=
\begin{cases}
  1, & \text{if } e \rightarrow v \\            
   -1, & \text{if } e \leftarrow v \\           
  0, & \text{otherwise}.
\end{cases}
\end{equation}
By using $A_{ev}$, we can write down equations of motion in a compact way, which will eventually lead us to our quantization prescription.

To see this, we define the branch charges $q_e$ only on the capacitive edges of the graph, whereas we assign node fluxes $\phi_v$ to all nodes. We consider circuits without time-dependent external fluxes and gate voltages for now. 

Let us first explain why, in fact, the flux variables should naturally live on vertices.  Kirchhoff's voltage rule states that in the absence of external magnetic fields around any loop of branches $e_1 \rightarrow e_2 \rightarrow \dots \rightarrow e_n \rightarrow e_1$ of (any) length $n$: \begin{equation}
\sum_{j=1}^n \phi_{e_j}=0.\label{eq:loopconstraint}
\end{equation}
Observe that if $\phi_e(t) = \phi_v(t)-\phi_u(t)$ whenever $e=u\rightarrow v$ [see Eq.~(\ref{eqn:branchflux})], then this constraint would automatically be obeyed.  Mathematically, one can actually prove that \emph{all} solutions to the constraints in Eq.~(\ref{eq:loopconstraint}) are of this form.\footnote{This uses homology theory, or the abstract form of Stokes' Theorem \cite{nakahara}.} 
It thus makes sense to think of the dynamical variables as the $\phi_v$, which are much less constrained, rather than $\phi_e$s, which obey all of the constraints in Eq.~(\ref{eq:loopconstraint}). 

Kirchhoff's current law states that for any node $v$, the exiting and entering currents are equal
\begin{equation}
    \sum_{e:e \text{ entering }v} I_e(t) = \sum_{e:e \text{ exiting }v} I_e(t). \label{eq:kirchoffIfirst}
\end{equation}
Using the incidence matrix $A_{ev}$ defined above, we can write Eq.~(\ref{eq:kirchoffIfirst}) as \begin{equation}
    \sum_e A_{ev}I_e(t) = 0. \label{eq:kirchoffI}
\end{equation}
  
To obtain a more explicit version of Kirchhoff's laws, we need to consider the energy of the various elements. The total inductive energy of the system is
\begin{equation}
 E_{\mathcal{I},\mathrm{tot}} = \sum_{e=u\rightarrow v \in \mathcal{I}} E_{\mathcal{I},e}(\phi_e) = \sum_{e=u\rightarrow v \in \mathcal{I}} E_{\mathcal{I},e}(\phi_v-\phi_u),
\end{equation}
where we have defined $E_{\mathcal{I},e}$ to be the energy function for the inductor on edge $e$. 
Similarly, there is energy stored in the capacitive branches: 
\begin{equation}
 E_{\mathcal{C},\mathrm{tot}} = \sum_{e\in \mathcal{C}} E_{\mathcal{C},e}(q_e),
\end{equation}
and the total energy of the system is
\begin{equation}
 E_\mathrm{tot} = E_{\mathcal{C},\mathrm{tot}} + E_{\mathcal{I},\mathrm{tot}}.
\end{equation} 
Here, we emphasize that the energies can be general functions of the charge branch variables or node flux variables. Note also that $E_{\mathrm{tot}}$ is not necessarily a Hamiltonian function, since the variables $\phi_v$ and $q_e$ are not conjugate in any obvious way.  We will explain how to obtain a Hamiltonian from $E_{\mathrm{tot}}$ by the end of the section.

Now, we describe the final equations of motion, which provide the glue between the charge and flux variables, and link time derivatives of $\phi_v$ and $q_e$ to the energies above. First, if $e=u\rightarrow v\in \mathcal{C}$ is a pair of nodes that are connected through a capacitive branch $e$, the voltage between the nodes equals to the voltage drop across the connecting capacitors.  Using (\ref{eq:capvoltage}), we find 
\begin{equation}\label{eq:eom_cap}
	\dot\phi_v - \dot \phi_u  - \frac{\partial E_{\mathcal{C},e}(q_e)}{\partial q_e} = 0 \mathrm{\ \ for\ all\ } e = u \rightarrow v\in \mathcal{C}.
\end{equation}
Second, at each node, we apply Kirchhoff's current law. For a capacitive edge we have $I_e = \dot{q}_e$, while for an inductive edge Eq.~(\ref{eq:indcurrent}) implies that \begin{equation}
    I_e - \frac{\partial E_{\mathcal{I},e}(\phi_v-\phi_u)}{\partial \phi} = 0  \mathrm{\ \ for\ all\ } e = u \rightarrow v\in \mathcal{I}.
\end{equation} 
Hence we arrive at our second equations of motion upon plugging in to Eq.~(\ref{eq:kirchoffI})
\begin{equation}\label{eq:eom_ind}
	\sum_{e=u^\prime\rightarrow v^\prime \in \mathcal{I}} A_{ev} \frac{\partial E_{\mathcal{I},e}(\phi_{v^\prime}-\phi_{u^\prime})}{\partial \phi} + \sum_{e\in\mathcal{C}} A_{ev} \dot q_e  =0 \mathrm{\ \ for\ all\ } v\in \mathcal{V}.
\end{equation}

Observe that these equations of motion follow from the following Lagrangian \begin{equation}
    L = - E_{\mathrm{tot}}(q_e \text{ for }e \in \mathcal{C},\phi_v \text{ for } v\in\mathcal{V}) + \sum_{e \in \mathcal{C}, v \in \mathcal{V}} q_e \Omega_{ev} \dot\phi_v, \label{eq:main}
\end{equation}
where $\Omega_{ev}$ is the incidence matrix whose rows correspond to \emph{capacitive edges only}, and columns correspond to \emph{all vertices}.  While entry-wise $\Omega_{ev}=A_{ev}$,  we use the $\Omega_{ev}$ notation to emphasize that now we only care about \emph{capacitive edges}.  $E_{\mathrm{tot}}$ is simply the energy of all circuit elements.  Indeed, Eq.~(\ref{eq:eom_cap}) comes from the Euler-Lagrange equation of motion of $\phi_v$, while Eq.~(\ref{eq:eom_ind}) comes from the equation of motion of $q_e$. This may seem like a miracle!  But we hope that by the end of this paper, the reader will walk away thinking that Eq.~(\ref{eq:main}) is in fact \emph{the most natural Lagrangian for a circuit}.  Firstly, it elegantly allows us to encode $\phi_v$ as degrees of freedom on nodes, while $q_e$ are degrees of freedom on branches.  Secondly, and much more importantly, it also encodes within it the universal prescription for circuit quantization.

As we describe in the upcoming sections, it is possible to remove non-dynamical degrees of freedom relying on the geometrical properties of the capacitive incidence matrix $\Omega_{ev}$. After removing those variables, we arrive at a set of $q_i$ and $\phi_j$ variables that are the linear combination of the original charge branch and flux node variables, where, crucially, the number of charge and flux variables are equal. With these variables, the Lagrangian reads
\begin{equation}
    L = - H(q_i, \phi_j) + \sum_{i,j} q_i \tilde\Omega_{ij} \dot\phi_j, \label{eq:main_reduced}
\end{equation}
where the total energy term corresponds to the Hamiltonian function $H(q_i, \phi_j)=E_\mathrm{tot}(q_i, \phi_j)$ and the second term in the Lagrangian contains a square invertible matrix $\tilde\Omega_{ij}$ that is linked to the symplectic matrix of the circuit. In fact, we will show in Section \ref{sec:tree} how to choose variables wherein $\tilde\Omega_{ij}$ is the identity matrix.

\subsection{From Hamiltonian mechanics to symplectic geometry}\label{sec:geom}
To understand how Eq.~(\ref{eq:main_reduced}) leads us to circuit quantization, we need to take a step back and comment on a crucial analogy.  Consider for the moment the classical mechanics textbook problem of an object with $N$ degrees of freedom with coordinates $(x^1,\ldots, x^N)$ and canonically conjugate momenta $(p^1,\ldots, p^N)$ that is described by a Hamiltonian $H(x^i,p^i)$.  In this section we will use raised or lowered indices to emphasize the connections with mathematics: raised indices correspond to vector fields on phase space, while lowered indices correspond to differential forms (this is historically known as contravariant vs. covariant vectors).  Observe that if we define the Lagrangian as \begin{equation}
    L = -H(x^i,p^i) + \sum_{i=1}^N p^i \dot x^i,
\end{equation}
and the action as $S=\int \mathrm{d}t \; L$, the Euler-Lagrange equations reproduce Hamilton's equations: \begin{subequations}\begin{align}
 0 &= \frac{\delta S}{\delta p^i} = \dot x^i - \frac{\partial H}{\partial p^i}, \\
 0 &= \frac{\delta S}{\delta x^i} = -\dot p^i - \frac{\partial H}{\partial x^i}.
\end{align}\end{subequations}

Now, suppose that we have an invertible matrix $M_{ij}$, and we define a ``Lagrangian" such that
\begin{equation}
    L = -H(x^i,p^i) + \sum_{i,j=1}^N p^i M_{ij} \dot x^j. \label{eq:LagwithM}
\end{equation}
Again we find a kind of Hamiltonian mechanics, but with a slightly modified form of Hamilton's equations.  Denoting the elements of the matrix inverse $M^{-1}$ with raised indices so that \begin{equation}
    \sum_{j=1}^N M^{ij}M_{jk} = \sum_{j=1}^N M_{kj}M^{ji} = \delta^i_k,
\end{equation}
we find that
\begin{subequations}
    \begin{align}
        0 &=\dot x^i- \sum_{j=1}^N M^{ij}\frac{\partial H}{\partial p^j}, \\
        0 &= -\dot p^i-\sum_{j=1}^N \frac{\partial H}{\partial x^j}M^{ji}.
    \end{align}
\end{subequations}
In simple terms, the $M$ matrix tells us that the canonical conjugate variables are not $p^i$ and $x^i$, but rather $p^i$ and $\sum_j M_{ij}x^j$. 

If we define the Poisson bracket for two functions $f$ and $g$ as\begin{equation}\label{eq:pbdefn}
    \lbrace f,g \rbrace = \sum_{i,j=1}^N M^{ij}\left(\frac{\partial f}{\partial x^i}\frac{\partial g}{\partial p^j} - \frac{\partial g}{\partial x^i}\frac{\partial f}{\partial p^j}\right),
\end{equation}
then Hamilton's equations can be re-written for an arbitrary function $f$ as \begin{equation}
    \dot f = \lbrace f,H\rbrace.  \label{eq:poissonbracketEOM}
\end{equation}
As we highlight in Section \ref{sec:qm}, such theories are straightforward to quantize.   Importantly, with a few caveats, our classical Lagrangian for a circuit, which is  $L = - H(q_i, \phi_j) + \sum_{i,j} q_i \tilde\Omega_{ij} \dot\phi_j$ has precisely the form of Eq.~(\ref{eq:LagwithM}).

It is helpful to re-formulate the previous paragraph in a more abstract language.   Let us collect the position and momentum coordinates into a single variable $\xi^I = (x^1,\ldots, x^N, p^1, \ldots, p^N)$, and define the matrix \begin{equation}
    \omega_{IJ} = \left(\begin{array}{cc} 0 &\ M_{ij} \\ -M_{ji} &\ 0 \end{array}\right).
\end{equation}
Note that $\omega_{IJ}=-\omega_{JI}$ is antisymmetric, invertible, in our special case a constant, and it is called the \textit{symplectic form}.\footnote{More formally, a symplectic form is required to be a closed, non-degenerate 2-form.  To the geometer, our $\omega_{IJ}$ easily satisfies these criteria.}  Mathematicians define Hamiltonian mechanics in terms of a Hamiltonian function $H$, which generates time evolution via the Poisson brackets in Eq.~(\ref{eq:poissonbracketEOM}), and a symplectic form $\omega$.   Defining the inverse of $\omega_{IJ}$ as $\omega^{IJ}$ as before: \begin{equation}
   \sum_{J=1}^{2N} \omega^{IJ}\omega_{JK} = \sum_{J=1}^{2N}\omega_{KJ}\omega^{JI} = \delta^I_K,
\end{equation}
we can re-write the Poisson bracket as \begin{equation}
   \lbrace f,g\rbrace = \sum_{I,J=1}^{2n} \frac{\partial f}{\partial \xi^I} \omega^{IJ} \frac{\partial g}{\partial \xi^J}.
\end{equation}

The pair of a manifold with coordinates $\xi^I$ and symplectic form $\omega$ is called a \textit{symplectic manifold}.  Such a symplectic manifold is required for a notion of Hamiltonian mechanics to exist \cite{arnold,dasilva}.  Importantly however, \emph{any} symplectic manifold gives rise to the structures of Hamiltonian mechanics -- even ones where there are no global canonical conjugate pairs of coordinates.  The mathematical theory of geometric quantization \cite{woodhouse} shows that there is a way to quantize all such systems.  

The key point is that ``Lagrangians" of the form of Eq.~(\ref{eq:LagwithM}), 
\begin{equation}
    L = -H(x^i,p^i) + \sum_{i,j=1}^N p^i M_{ij} \dot x^j,
\end{equation}
which are similar to our circuit Lagrangian in Eq.~(\ref{eq:main}), 
\begin{equation}
    L = - H(q_i, \phi_j) + \sum_{i,j} q_i \tilde\Omega_{ij} \dot\phi_j, 
\end{equation}
are immediately understood in the language of Hamiltonian mechanics and symplectic geometry.  In particular, we simply \emph{read out the Hamiltonian function $H(q_i, \phi_j)$ and a Poisson bracket $ \{\phi_i, q_j\} = \tilde\Omega_{ji}^{-1}$}, which allows us to elegantly transition from classical to quantum mechanics.  

\subsection{Symplectic geometry of a circuit}\label{sec:symplectic_geometry}

After having reviewed the basics of symplectic geometry, we return to the question of how to construct the Hamiltonian $H(q_i,\phi_j)$ function of an arbitrary circuit from its total energy $E_\mathrm{tot}(q_e,\phi_v)$ and the connectivity of the elements $\Omega_{ev}$. In this section, we focus on the general approach, while in Sec.~\ref{sec:examples} we provide examples. A mathematically precise discussion, with proofs of all claims, is relegated to Appendix~\ref{app:math}.  

To begin, it is important to notice that the incidence matrix $\Omega_{ev}$ appearing in the Lagrangian of the circuit [see Eq.~(\ref{eq:main})] is not invertible. This is in contrast to the definition of the symplectic matrix, which is constructed from an invertible matrix $M_{ij}$ [see Eq.~(\ref{eq:LagwithM})]. Thus, at this point, it is not possible to carry out a Legendre transformation to arrive from the Lagrangian to a Hamiltonian with conjugate flux and charge pairs. The root of the problem is that we overcounted the degrees of freedom the way we constructed the Lagrangian. However, as we prove in Appendix~\ref{app:math} and summarize in Section \ref{sec:tree}, we can consistently and efficiently remove variables associated with constraints to obtain an invertible matrix (and a symplectic form) from the incidence matrix. Crucially, this procedure only depends on the geometrical structure of the capacitive subgraph of the circuit: the locations of inductively shunted islands and capacitive loops (see Fig.~\ref{fig:nullvectors}). 

There are three general methods to reduce the number of variables in our approach, which we highlight here (and give examples of in Sec.~\ref{sec:examples}).
\begin{figure}
    \centering
    \includegraphics[width = 8.6cm]{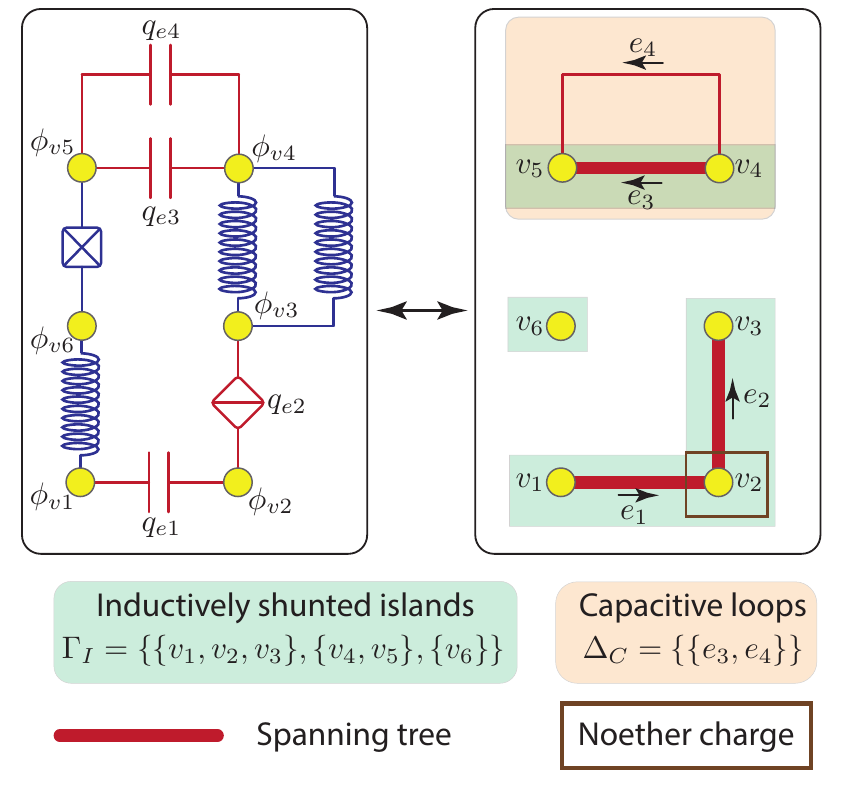}
    \caption{\textbf{Symplectic geometry of a circuit.} When describing a circuit with its capacitive graph, we can define two types of geometrical objects that correspond to null vectors of the capacitive incidence matrix of the circuit: inductively shunted islands (green-filled rectangular), and capacitive loops (orange-filled rectangular). The variables associated with these null vectors must be removed to be able to quantize the circuit. Additional variables can be removed based on the Noether charges of the circuit. The edges that are part of a (non-unique) spanning tree are highlighted with wide lines.}
    \label{fig:nullvectors}
\end{figure}
Firstly, we may find a left null vector, $l_e$, of the incidence matrix: a linear combination of the branches such that
\begin{equation}
\sum_e l_e\Omega_{ev}=0.
\end{equation}
Geometrically, these null vectors correspond to loops in the circuit, where all branches in a loop have capacitive elements on them (see proof in Appendix \ref{app:math}). Physically, the Euler-Lagrange equations for these null vectors lead to the physical constraint that the voltages in such a \textbf{capacitive loop} vanishes
\begin{equation}
\sum_e l_e\frac{\partial E_\mathcal{C}}{\partial q_{e}}=0.
\label{eq:constraint_left}
\end{equation}
This constraint fixes one of the charge variables in terms of the others in the loop. We denote the set of such capacitive loops as $\Delta_C$. For a loop $Z\in\Delta_C$, the form of the null vectors is \begin{equation}
    l_e = \left\lbrace \begin{array}{ll} \pm 1 &\ e \in Z \\ 0 &\ e\notin Z \end{array}\right..
\end{equation}
The $\pm 1$ sign is based on the orientation of the edges in the loop: all signs are $+1$ when the edges are all oriented so they touch tip-to-tail: see the examples for more details in Sec.~\ref{sec:examples}.

Secondly, the right null vectors $r_v$ of $\Omega_{ev}$ also imply constraints. These are the combinations of nodes such that
\begin{equation}
\sum_v \Omega_{ev} r_v=0.
\end{equation}
The geometrical meaning of these vectors is that they represent \textbf{inductively shunted islands}. The constraint associated with these vectors is that the total current entering the island must equal the current exiting the island. In this case, the Euler-Lagrange equations read
\begin{equation}
\sum_v \frac{\partial E_\mathcal{I}}{\partial \phi_{v}}r_v=0.
\label{eq:constraint_right}
\end{equation}
We denote the set of all subsets of vertices that correspond to such inductively shunted islands as $\Gamma_I$. Note that for any such island $J\in \Gamma_I$, we have $J\subseteq \mathcal{V}$.  The explicit form of the null vector $r_v$ becomes \begin{equation}
    r_v = \left\lbrace\begin{array}{ll} 1 &\ v\in J \\ 0 &\ v\notin J\end{array}\right..
\end{equation}

These left and right null vectors correspond to non-dynamical variables that can be removed from the Lagrangian.
After removing the variables associated with the left and right null vectors of $\Omega_{ev}$, the Lagrangian will have fewer coordinates.  The linear combinations of coordinates which remain are the non-null vectors of $\Omega_{ev}$, which becomes a non-degenerate matrix $\tilde\Omega_{ij}$ in the subspace of remaining modes.  The total energy $E_{\mathrm{tot}}$, restricted to the corresponding constrained subspace, is the Hamiltonian $H$ for the circuit.  Hence, we can find a symplectic form for the remaining coordinates, and a Hamiltonian function to quantize.  
 
 At this point, we can formally quantize the theory, as we describe in Sec.~\ref{sec:qm}, by replacing Poisson brackets with quantum commutators.  The two steps above are thus \emph{required} prior to quantization, but there is another way to easily remove some degrees of freedom.  Suppose that, as in Fig.~\ref{fig:nullvectors}, there is a vertex (or more generally, a set of vertices) that are only connected to the rest of the circuit via capacitive edges.  For simplicity here, let us focus on the case where, as in Fig.~\ref{fig:nullvectors}, it is a single vertex $v_2$.  Then, the Lagrangian $L$ in Eq.~(\ref{eq:main}), and therefore the Hamiltonian $H$, is invariant under constant shifts in $\phi_{v_2}$:
 \begin{equation}
    H( \phi_{v_2}) = H( \phi_{v_2} + c).
 \end{equation}
 Noether's Theorem states that such continuous symmetry allows one to remove one dynamical variable (i.e. one $q$ and one $\phi$) from the problem.\footnote{More formally, there exists a quotient of the classical phase space (intuitively amounting to ignoring the dynamics of the conserved quantity and its conjugate) which remains a symplectic manifold. In the mathematical literature, this is known as the Marsden-Weinstein-Meyer Theorem, and the method of symplectic reduction.}
 
 \subsection{Choosing canonically conjugate variables}\label{sec:tree}

As we will see when quantizing the theory, it is desirable to find $n$ pairs of ``canonical coordinates": \begin{equation}\label{eq:trivialcom}
    \{x_i, p_j\} = \delta_{ij}.
\end{equation}
This is because, in quantum mechanics, Poisson brackets become quantum mechanical commutators.
However, when calculating the Poisson brackets using Eq.~(\ref{eq:pbdefn}) in our formalism, we find that the Poisson brackets correspond to the element of the symplectic matrix
\begin{equation}
    \{\phi_j, q_i\} = \tilde\Omega_{ij}^{-1}.
\end{equation}
This does not jeopardize our ability to quantize the circuit, but it is still desirable to find coordinates where Eq.~ (\ref{eq:trivialcom}) holds.  
In this section, we show how to find charge and flux variables that achieve this. For simplicity, we will focus on an example presented in Fig.~\ref{fig:nullvectors}, and relegate the general argument to  Appendix \ref{app:math}.

Our argument is exclusively about the second term in the Lagrangian of the system in Eq.~(\ref{eq:mainintro}), which reads as  $\sum_{e,v} q_e \Omega_{ev} \dot\phi_v$. We aim to find a set of $(\Phi_i,Q_i)$ variables for which this term takes the form of $\sum_i Q_i \dot\Phi_i$. As we discussed before in Sec.~\ref{sec:symplectic_geometry}, in part this will mean removing all left and right null vectors.  Remarkably, in the construction that follows, these null vectors will be automatically removed.

First, we choose a spanning tree $\mathcal{T} \subseteq \mathcal{C}$ of the capacitive subgraph.  Here a spanning tree corresponds to a set of capacitive branches $\mathcal T \subset \mathcal C$ so that every node adjacent to some branch in $\mathcal C$ is adjacent to some branch in $\mathcal T$, but without any cycles. Schematically, one can manufacture a spanning tree by simply choosing an edge to delete from every cycle in $\Delta_C$.  For example in Figure \ref{fig:nullvectors}, we can take the spanning tree to be \begin{equation}
     \mathcal{T} = \lbrace e_1,e_2,e_3\rbrace.
 \end{equation}
 An alternative choice is $\lbrace e_1,e_2,e_4\rbrace$, and the choice made does not affect the spectrum or dynamics of the resulting circuit (the resulting Hamiltonians differ by a canonical transformation).
 Recalling the definition of branch flux in Eq.~(\ref{eqn:branchflux}), we define branch fluxes $\Phi_f$ for $f\in \mathcal{T}$ as our fundamental degrees of freedom.  One can explicitly show that \begin{equation}
   \sum_{e\in\mathcal{C}}\sum_{v\in \mathcal{V}} q_e \Omega_{ev}\dot\phi_v = \sum_{f\in\mathcal{T}} Q_f \dot\Phi_f,
\end{equation}
where $Q_f$ is a linear combination of the original $q_e$ variables with integer coefficients $0, \pm 1$.  In our example, we find \begin{equation}
     \sum_{e\in\mathcal{C}}\sum_{v\in \mathcal{V}} q_e \Omega_{ev}\dot\phi_v = q_{e_1}\dot\Phi_{e_1} + q_{e_2}\dot\Phi_{e_2} + (q_{e_3}+q_{e_4}) \dot\Phi_{e_3} = Q_{e_1}\dot\Phi_{e_1}+Q_{e_2}\dot\phi_{e_2}+Q_{e_3}\dot\Phi_{e_3}. \label{eq:sec25ex}
\end{equation}We will show how to use this procedure in the additional examples of Sec.~\ref{sec:examples}.  

Observe that in this construction, we have immediately removed one linear combination of node fluxes on each disconnected subgraph of $\mathcal{C}$.  In Fig.~\ref{fig:nullvectors}, we can see the following right null vectors are no longer dynamical degrees of freedom:  $\phi_{v_6}$, $\phi_{v_1}+\phi_{v_2}+\phi_{v_3}$, $\phi_{v_4}+\phi_{v_5}$.  The three branch flux variables on the spanning tree are linearly independent to these non-dynamical modes.  Similarly, the linear combinations of charges that are removed are the unphysical ones corresponding to charges flowing around capacitive loops.  In our example, this combination of branch charges $q_{e_3}-q_{e_4}$  is orthogonal to the physical degree of freedom $Q_{e_3} = q_{e_3}+q_{e_4}$ that arose in Eq.~(\ref{eq:sec25ex}).   Happily, all the needed left and right null  vectors of $\Omega_{ev}$ are automatically removed by this ``spanning tree construction" of choosing good coordinates.  

\section{Circuit quantization}\label{sec:qm}
With the understanding of how to use our formalism to describe  arbitrary non-dissipative circuits at the classical level, we now discuss how to carry out circuit quantization. Suppose that the circuit is described by the Lagrangian
\begin{equation}\label{eqn:qlag}
    L = -H(q_i,\phi_j)+\sum_{i,j}  q_i \tilde\Omega_{ij} \dot \phi_j = -H(Q_f,\Phi_f) + \sum_{f\in\mathcal{T}}Q_f \dot\Phi_f.
\end{equation}
Note that we have used the spanning tree construction of Sec.~\ref{sec:tree} to choose good variables to quantize. The equations of motion for the non-dynamical coordinates are constraints that should also be solved before quantization.

Referring to the definition of the Poisson brackets in Eq.~(\ref{eq:pbdefn}), we see that 
\begin{equation}
    \{\Phi_i,Q_j\} = \delta_{ij}.
\end{equation}
To quantize the circuit, we define commutation relations between the charge operator $\hat{Q_i}$ and the flux operator $\hat{\Phi_j}$ as
\begin{equation}
    [\hat\Phi_i,\hat{Q}_j] = \mathrm i \hbar \delta_{ij},\label{eq:quantization}
\end{equation}
as long as both $\hat{Q_i}$ and $\hat{\Phi_j}$ are non-compact (i.e., not periodically identified) variables.  

The quantum mechanical Hamiltonian is simply $\hat{H}(\hat Q_i,\hat \Phi_i)$, where as in the classical setting, we must first restrict to the constrained subspace by solving for left/right null vectors of $\Omega_{ev}$.  Since in our theory, all circuit elements are purely capacitive or purely inductive, there is no ambiguity about the operator ordering of non-commuting $\hat Q_i$ and $\hat \Phi_i$, so $\hat{H}$ is a uniquely specified operator.  This completes our formulation of circuit quantization for non-dissipative circuits.

Such a simple and intuitive solution to circuit quantization is possible because we are able to find a globally constant Poisson bracket on the classical phase space. There do exist Hamiltonian systems where this task cannot be achieved.\footnote{A simple example is a system whose classical phase space is the sphere.  This system is quantizable and gives a single spin-$j$ Hilbert space of dimension $2j+1$ \cite{woodhouse}, when the volume of the sphere is $2\pi\hbar (2j+1)$.}
The most notable property of our quantization procedure, and our formalism on the whole, is that Eq.~(\ref{eq:quantization}) is agnostic to the form of the Hamiltonian; it depends only on the capacitive subgraph of the circuit.

When considering circuits, it is often the case that some of the flux coordinates are periodically identified: \begin{equation}
    \Phi_i \sim \Phi_i + \phi_0, \label{eq:periodicphi}
\end{equation}
where $\phi_0$ is the flux quantum. For example, it is generally assumed that flux across a Josephson junction shunted by a capacitor is periodic\footnote{However, there are other approaches that assume no periodicity in this case~\cite{thanh_le_building_2020}.}.  It is well-known \cite{periodic} that in such circuits, $\Phi_i$ is not a well-defined operator; the well-defined operators become $\exp[2\pi \mathrm{i}\Phi/\phi_0 \cdot n]$ for integer $n$.  Our quantization prescription does not change in this scenario: one simply avoids writing Eq.~(\ref{eq:quantization}) and instead writes \begin{equation}\label{eq:expcommutator}
    [\mathrm{e}^{2\pi \mathrm{i} \hat\Phi/\phi_0},\hat{Q}/2e] = -\mathrm{e}^{2\pi \mathrm{i}\hat\Phi/\phi_0},
\end{equation}
which is now expressed in terms of globally-defined operators. 

Let us remark on what has transpired from a mathematical perspective.  For simplicity let us assume a single dynamical $Q$ and $\Phi$ variables; the argument immediately generalizes to the higher dimensional case. The original classical phase space is $M=\mathbb{R}^{2}$.  The periodic identification of $\Phi$ corresponds to identifying points in phase space when the $\Phi$ coordinates are related as in Eq.~(\ref{eq:periodicphi}). At the classical level, this turns the phase space into $\mathbb{R}\times \mathrm{S}^1$, where $Q\in \mathbb{R}$ and $\Phi \in \mathrm{S}^1$.  Here $\Phi \in \mathrm{S}^1$ lives on a circle, which is equivalent to the real line with all points shifted by $\phi_0$ identified.  Because the manifold $\mathbb{R}\times \mathrm{S}^1$ is a non-singular quotient of $\mathbb{R}^2$, there exists \cite{manifold} a symplectic form $\omega$ on $\mathbb{R}\times \mathrm{S}^1$ which is equal to the inclusion of the original symplectic form on $\mathbb{R}^2$.  In more physical terms, this means that we can use the same commutation relations to quantize the reduced phase space, provided we only study well-defined functions as in Eq.~(\ref{eq:expcommutator}).  Note that in quantum mechanics, $\Phi$ becoming periodic means that $Q$ becomes integer-valued; this has no classical analogue, and the theory of geometric quantization was developed to explain this phenomenon for general symplectic manifolds \cite{woodhouse}.  In this paper, we will be dealing with classical phase spaces that are quotients of $\mathbb{R}^{2n}$ by periodically identifying $\Phi$ coordinates, so these subtleties will end up unimportant.

\section{Examples and generalizations}\label{sec:examples}

In this section, we provide examples of how to use our formalism to efficiently derive a quantizable Hamiltonian for various circuits. 

\subsection{Inductively and capacitively shunted islands}

As a first example to understand how we can eliminate variables associated with unphysical degrees of freedom, we consider an inductively shunted island. An inductively shunted island contains a set of nodes that lie on a path consisting of only capacitive branches or a single node that is connected only to inductive elements. For formal definitions and other relevant discussions, see Appendix \ref{app:math}. As we discussed before, an inductively shunted island corresponds to a right null vector of the incidence matrix $\Omega_{ev}$. Figure \ref{fig:examples_1}(a) shows an example of a circuit that has two inductively shunted islands: a node is connected to two inductors, $L_1$ and $L_2$, and a quantum phase slip element with energy $E_Q$ is also shunted by the inductors. The circuit has three node variables but using the constraints for right null vectors [see Eq.~(\ref{eq:constraint_right})], we can eliminate two flux variables.

\begin{figure}
    \centering
    \includegraphics[width = 8.6cm]{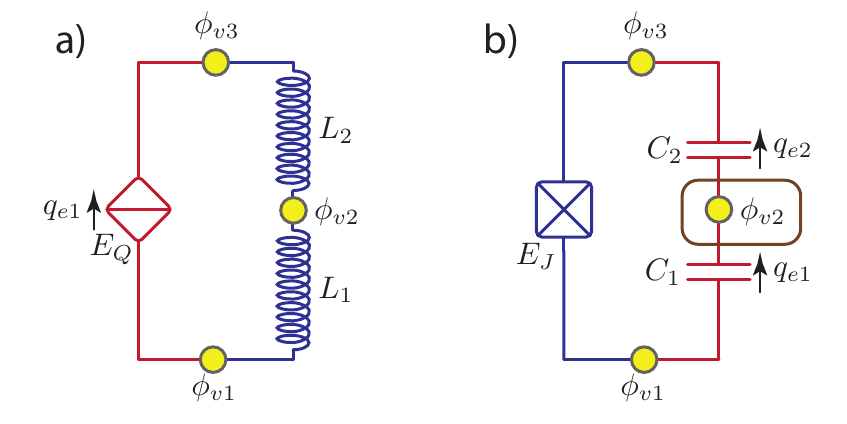}
    \caption{\textbf{Inductively and capacitively shunted islands.} (a) The inductively shunted island at node flux $\phi_{v_2}$ corresponds to the right null vector of the incidence matrix $\Omega_{ev}$; thus, we need to remove such variable to be able to arrive to a self-consistent Hamiltonian. (b) The capacitively shunted island at $\phi_{v_2}$ corresponds to a Noether charge in the circuit. It is possible but not required to remove such a variable to be able to define conjugate pairs.}
    \label{fig:examples_1}
\end{figure}

To start, we write down the incidence matrix of the capacitive subgraph
\begin{equation}
\Omega = 
\begin{pmatrix}
-1 & 0 & 1
\end{pmatrix},
\end{equation}
where the single row corresponds to the $q_{e_1}$ branch charge, and the three columns refer to the three flux node variables. The capacitive and the inductive energies are 
\begin{subequations}
\begin{align}
E_\mathcal{C} &= -E_Q\cos\left(2\pi\frac{q_{e_1}}{2e}\right), \\
E_\mathcal{I} &= \frac{1}{2L_1}\left(\phi_{v_1}-\phi_{v_2}\right)^2 + \frac{1}{2L_2}\left(\phi_{v_2}-\phi_{v_3}\right)^2.
\end{align}
\end{subequations}
Thus, based on Eq.~(\ref{eq:mainintro}) the Lagrangian is
\begin{equation}
\begin{aligned}
  L &= \sum_{e,v}q_e\Omega_{ev}\dot{\phi}_v - E_\mathcal{I} - E_{\mathcal{C}}  \\ 
  &= q_{e_1}(\dot\phi_{v_3}-\dot\phi_{v_1}) + E_Q\cos\left(2\pi\frac{q_{e_1}}{2e}\right) - \frac{1}{2L_1}\left(\phi_{v_1}-\phi_{v_2}\right)^2 - \frac{1}{2L_2}\left(\phi_{v_2}-\phi_{v_3}\right)^2.
  \end{aligned}
\end{equation}

We notice that there are two inductively shunted islands and hence two right null vectors
\begin{equation}
\Gamma_I
= \{\{v_1, v_3\},\{v_2\}\}\longleftrightarrow r_v = \begin{pmatrix}
    1 \\
    0 \\
    1
\end{pmatrix},
\begin{pmatrix}
    0 \\
    1 \\
    0
\end{pmatrix},
\end{equation}
indicating that $\phi_{v_1}+\phi_{v_3}$ and $\phi_{v_2}$ are non-dynamical variables. Furthermore, 
based on the constraints imposed by the right null vectors\footnote{The constraints provided by each of the two null vectors are the same in this example, but need not be in general.} [see Eq.~(\ref{eq:constraint_right})], we can write that
\begin{equation}\label{eq:nontrivcons}
    \frac{\phi_{v_1} -\phi_{v_2}}{L_1} +  \frac{\phi_{v_3} - \phi_{v_2}}{L_2} = 0.
\end{equation}

After some algebra, we can simplify the Lagrangian  such as 
\begin{equation}
L = Q\dot{\Phi} + E_Q\cos\left(2\pi\frac{Q}{2e}\right) - \frac{1}{2(L_1+L_2)}\Phi^2, 
\end{equation}
where $\Phi=\phi_{v_3} - \phi_{v_1}$, and $Q=q_{e_1}$. We can see that this geometrical method reproduced the well-known result of how to add inductors together. The system is left with one degree of freedom, and the symplectic form (the first term in the Lagrangian)  indicates that the conjugate variables are $\{\Phi,Q\}=1$. Finally, the Hamiltonian is
\begin{equation}
H = - E_Q\cos\left(2\pi\frac{Q}{2e}\right) + \frac{1}{2(L_1+L_2)}\Phi^2. 
\end{equation}
We remark that the conjugate pairs in this example were necessarily $Q$ and $\Phi$ because there is only one capacitive branch and thus it must have been included in any spanning tree.

 As a second example, we consider a capacitively shunted island, for example, a node between two capacitors [see Fig.~\ref{fig:examples_1}(b)]. A capacitively shunted island is a set of vertices that can be traversed by moving only along branches with inductive elements. As before, a node connected only to capacitors constitutes its own island. In our formalism, the presence of such an island does not lead to a null vector of the adjacency matrix. Thus, removing such variables is not necessary to define a symplectic form and to carry out quantization (see Appendix \ref{app:noether} for an example). However, we can remove such degrees of freedom since capacitively shunted islands correspond to Noether currents, which represent an additional constraint.  
Physically, this constraint corresponds to Kirchhoff's current law, i.~e., the current through a network of capacitors is conserved. 

In this example [see Fig.~\ref{fig:examples_1}(b)], the circuit contains a Josephson junction and two capacitors in a loop; the Lagrangian describing this circuit is given by Eq.~(\ref{eq:mainintro})
\begin{equation}
    L = q_{e_1}(\dot \phi_{v_2} - \dot \phi_{v_1} ) + q_{e_2} (\dot\phi_{v_3} - \dot \phi_{v_2} ) + E_J \cos\left(2 \pi\frac{\phi_{v_1} - \phi_{v_3}}{\phi_0}\right) - \frac{1}{2 C_1} q_{e_1}^2 - \frac{1}{2 C_2} q_{e_2}^2. 
\end{equation}
Our spanning tree\footnote{We remark that, again, the choice of spanning tree on this circuit is unique.} construction provides for us the fact that the variables
\begin{equation}
\begin{aligned}
    Q_{e_1} &= q_{e_1}, \\ 
    Q_{e_2} &= q_{e_2}, \\ 
    \Phi_{e_1} &= \phi_{v_2} - \phi_{v_1}, \\ 
    \Phi_{e_2} &= \phi_{v_3} - \phi_{v_2} 
    \end{aligned}
\end{equation}
form canonical conjugate pairs with 
\begin{equation}
    \{\Phi_{e'},Q_{e} \} = \delta_{ee'}.
\end{equation}
Written in terms of these variables, the Lagrangian is
\begin{equation}
    L = Q_{e_1} \dot \Phi_{e_1 } + Q_{e_2} \dot \Phi_{e_2} + E_J\cos\left(2 \pi\frac{\Phi_{e_1} + \Phi_{e_2}}{\phi_0} \right) - \frac{1}{2 C_1} Q_{e_1}^2 - \frac{1}{2 C_2}Q_{e_2}^2 .
\end{equation}
At this point, the circuit can be quantized. However, we can remove one more variable by noticing that the constraint due to the Noether current is
\begin{equation}\label{eqn:smallex}
    \frac{\delta S }{\delta \phi_{v_2} } = \dot q_{e_1} - \dot q_{e_2}= \dot Q_{e_1} - \dot Q_{e_2} = 0.
\end{equation}
This is easy to understand as simply the conservation of charge on the two inner plates connecting the capacitors.
A nonzero constant of integration would only represent a time-independent charge trapped between the plates therein. 
In this case, we can write 
\begin{equation}
    Q_{e_1} \dot \Phi_{e_1} + Q_{e_2} \dot\Phi_{e_2} = Q(\dot\Phi_{e_1} + \dot\Phi_{e_2}),
\end{equation}
where we redefine $Q = Q_{e_1} = Q_{e_2} $ and $\Phi = \Phi_{e_1}+\Phi_{e_2}  $  with $\{\Phi,Q\} = 1$. 
This choice is just a reflection of the fact that a free particle is ``integrated out" by using Eq.~(\ref{eqn:smallex}). Thus, the Hamiltonian reads
\begin{equation}
    H = \frac{1}{2}\left(\frac{1}{C_1} + \frac{1}{C_2}\right)Q^2- E_J\cos\left(2 \pi\frac{\Phi}{\phi_0} \right)  .
\end{equation}
In this way, we see that Noether charges provide instructions on how to add capacitive circuit elements in series. 

\subsection{The dualmon qubit}

We continue the series of examples with the circuit that motivated our discussion, the dualmon circuit~\cite{le_doubly_2019}. In this device, a Josephson junction and a quantum phase slip element form a loop [see Fig.~\ref{fig:examples}(a)]. In the following, we analyze the circuit in the absence of offset charges and external fluxes. Later, we show how these external parameters can be added to our formalism.

First, we define the flux variables at the two nodes, $\phi_{v_1}$ and $\phi_{v_2}$, and the branch charge across the capacitive element, $q_{e_1}$. The incidence matrix of the capacitive subgraph is simply 
\begin{equation}
\Omega = 
\begin{pmatrix}
-1 & 1
\end{pmatrix}.
\end{equation}

The circuit has one inductively shunted island, and no capacitive loop, thus the nullvectors are
\begin{subequations}
\begin{align}
\Gamma_I = \{\{v_1, v_2\}\}&\longleftrightarrow r_v = \begin{pmatrix}
    1 \\
    1
\end{pmatrix} \\
\Delta_C = \{\}&\longleftrightarrow l_e = \begin{pmatrix}
    0
\end{pmatrix}.
\end{align}
\end{subequations}
Based on Eq.~(\ref{eq:constraint_right}), the constraint arising from the right null vector is trivial, and does not reduce the number of variables in the circuit. However, the identification of the right null vector itself formally removes a degree of freedom, which can also be seen in that the variable $\phi_{v_1} + \phi_{v_2}$ never appears in the Lagrangian. 
Formally, this null vector is appropriately removed by the spanning tree construction.

From the spanning tree construction, we see that $Q = q_{e_1}$ is conjugate to $\Phi = \phi_{v_2} - \phi_{v_1} $ so that $\{\Phi,Q\} = 1$.
Furthermore, if the Josephson energy is $E_J$ and the quantum phase energy is $E_Q$, the capacitive and inductive energies in the circuit are
\begin{subequations}
\begin{align}
E_\mathcal{C} &= -E_Q\cos\left(2\pi\frac{Q}{2e}\right), \\
E_\mathcal{I} &= -E_J\cos\left(2\pi\frac{\Phi}{\phi_0}\right).
\end{align}
\end{subequations}
Using Eq.~(\ref{eq:mainintro}), we write the Lagrangian of the circuit as
\begin{equation}
\begin{aligned}
  L &= \sum_{e,v}q_e\Omega_{ev}\dot{\phi}_v - E_\mathcal{I} - E_{\mathcal{C}}  \\ 
  &= Q\dot \Phi + E_Q\cos\left(2\pi\frac{Q}{2e}\right) + E_J\cos\left(2\pi\frac{\Phi}{\phi_0}\right).
  \end{aligned}
\end{equation}
 
Finally, the Hamiltonian function takes the form of
\begin{equation}
  H = -E_Q\cos\left(\frac{2\pi}{2e}Q\right) - E_J\cos\left(\frac{2\pi}{\phi_0}\Phi\right).
\end{equation}

\begin{figure}
    \centering
    \includegraphics[width = \columnwidth]{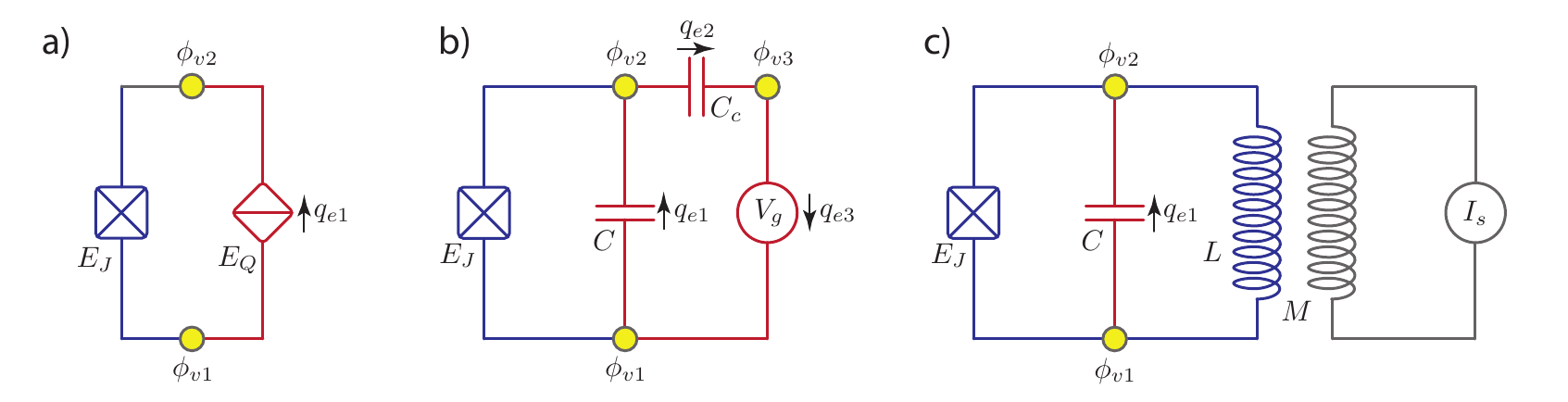}
    \caption{\textbf{Examples for quantum circuits in the framework of symplectic geometry.} The branches are colored based on the type of element they contain; red: capacitive elements, blue: inductive elements. The circuits are (a) dualmon circuit, (b) offset-charge-sensitive transmon, (c) external-flux-sensitive fluxonium.} 
    \label{fig:examples}
\end{figure}

\subsection{Offset-charges and external voltages in the transmon}

Now, we show how we can incorporate the offset charges in our description through the example of the offset-charge sensitive transmon or Cooper pair box [see Fig.~\ref{fig:examples}(b)]. The circuit contains a single Josephson junction with Josephson energy of $E_J$ shunted by a capacitor $C$, and coupled with capacitance $C_c$ to a classical gate voltage $V_g$ that models the effects of offset charges.  A key observation is that \emph{in our formalism, we include voltage sources by treating them as additional capacitive edges.} While they do not end up leading to new degrees of freedom, this is how they are straightforwardly handled in our framework.

In our example, the circuit has three nodes ($v_i$, where $i=1,2,3$), 
 and three capacitive branches, including the voltage source ($e_i$, where $i=1,2,3$). Thus, the capacitive incidence matrix is
\begin{equation}
\Omega = 
\begin{pmatrix}
  -1 & 1 & 0 \\
  0 & -1 & 1  \\
  1 & 0 & -1  
\end{pmatrix},
\end{equation}
where the columns correspond to the three vertices and rows to the three branches. By inspection, we note that the inductively shunted islands, capacitive loops, and the corresponding null vectors in the circuit are
\begin{subequations}
\begin{align}
\Gamma_I = \{\{v_1, v_2, v_3\}\} &\longleftrightarrow r_v = \begin{pmatrix}
    1 \\
    1 \\
    1
\end{pmatrix}, \\
\Delta_C = \{\{e_1,e_2,e_3\}\}
 &\longleftrightarrow l_e = \begin{pmatrix}
    1 & 1 & 1
\end{pmatrix}. 
\end{align}
\end{subequations}
The capacitive and inductive energies are
\begin{subequations}
\begin{align}
E_\mathcal{C} &= \frac{q_{e_1}^2}{2C} + \frac{q_{e_2}^2}{2C_c} +q_{e_3}V_g, \\
E_\mathcal{I} &= -E_J\cos\left(2\pi\frac{\phi_{v_2}-\phi_{v_1}}{\phi_0}\right).
\end{align}
\end{subequations}

Thus, based on Eq.~(\ref{eq:mainintro}) the Lagrangian of the circuit reads
\begin{equation}
\begin{aligned}
  L &= \sum_{e,v}q_e\Omega_{ev}\dot{\phi}_v - E_\mathcal{I} - E_{\mathcal{C}}  \\ 
  &= q_{e_1}(\dot\phi_{v_2}-\dot\phi_{v_1}) + q_{e_2}(\dot\phi_{v_3} - \dot\phi_{v_2}) + q_{e_3} (\dot\phi_{v_1} - \dot\phi_{v_3}) - \frac{q_{e_1}^2}{2 C}  - \frac{q_{e_2}^2}{2 C_c}    - q_{e_3} V + E_J\cos\left[\frac{2\pi}{\phi_0}(\phi_{v_2}-\phi_{v_1})\right].
  \end{aligned}
\end{equation}

In this example, the choice of spanning tree is not unique. We will choose 
\begin{equation}
   \mathcal T = \{e_1, e_2\} 
\end{equation}
as a spanning tree, and because the sum of the branch fluxes in the loop vanishes
\begin{equation}
    \phi_{e_3} = - \phi_{e_1}- \phi_{e_2}.
\end{equation}
Then we introduce the new variables
\begin{equation}
    \begin{aligned}
        Q_{e_1} &= q_{e_1} - q_{e_3}, \\ 
        Q_{e_2} &= q_{e_2} - q_{e_3}, \\
        \Phi_{e_1} &= \phi_{e_1}, \\ 
        \Phi_{e_2} &= \phi_{e_2}. 
    \end{aligned}
\end{equation}
We are free to rewrite 
\begin{equation}
    L = Q_{e_1} \dot\Phi_{e_1}  + Q_{e_2} \dot\Phi_{e_2} - \frac{(Q_{e_1} + q_{e_3})^2}{2 C} - \frac{(Q_{e_2} + q_{e_3} )^2 }{2 C_c} - q_{e_3} V_g + E_J \cos\left[ \frac{2\pi}{\phi_0} \Phi_{e_1}\right].
\end{equation}

Only the capacitive loop (left null vector) gives a nontrivial constraint based on Eq.~(\ref{eq:constraint_left}) since
\begin{equation}
  \frac{Q_{e_1} + q_{e_3}}{C} + \frac{Q_{e_2}+q_{e_3}}{C_c} + V_g = 0.
\end{equation}
which can be used to fix $q_{e_3}$ in terms of $Q_{e_1}$ and $Q_{e_2}$.
Further, there is a Noether current which produces the constraint 
\begin{equation}
    0 = \frac{\delta S }{\delta \phi_{v_3} } = \dot q_{e_2} - \dot q_{e_3} = \dot Q_{e_2} = 0.
\end{equation}

Choosing the constant of integration to be zero, and defining $Q= Q_{e_1}$ and $\Phi = \Phi_{e_1}$ with $\{\Phi,Q\}=1$, we arrive at the Lagrangian in the form of
\begin{equation}\label{eq:L430}
L= Q\dot{\Phi} - \frac{(Q-C_cV_g)^2}{2(C+C_c)} + E_J\cos\left(2\pi\frac{\Phi}{\phi_0}\right),
\end{equation}
after dropping a constant term.   
Thus, we can write the Hamiltonian function in the well-known form
\begin{equation}
H = \frac{(q-C_cV_g)^2}{2(C+C_c)} - E_J\cos\left(2\pi\frac{\Phi}{\phi_0}\right)
\end{equation}
From this point, it is straightforward to quantize $H$ even with compact variable $\Phi$ [see Eq.~(\ref{eq:expcommutator})].

\subsection{External flux in the fluxonium}

Now, we turn our attention to the case of external fluxes. It is generally straightforward to include external flux biases; here, we model it by coupling the circuit inductively to a loop with current $I_s$  flowing [see Fig.~\ref{fig:examples}(c)]. If the mutual induction is $M$, the relevant energy term is 
\begin{equation}\label{eqn:biasI}
    E = M I_s (\phi_{v_1} - \phi_{v_2}). 
\end{equation}

For the sake of brevity, we will simply write out the Lagrangian of the fluxonium following similar procedures as in the first two examples:
\begin{equation}
L = q_{e_1}(\dot\phi_{v_2}-\dot\phi_{v_1}) +\frac{q_{e_1}^2}{2C} + E_J\cos\left(2\pi\frac{\phi_{v_2}-\phi_{v_1}}{\phi_0} \right) -\frac{1}{2L}(\phi_{v_2}-\phi_{v_1})^2 -M I_s (\phi_{v_2}-\phi_{v_1}),
\end{equation}
which, after finding the spanning tree, can be further written as
\begin{equation}
    L = Q\dot \Phi + \frac{Q^2}{2C} + E_J \cos \left(2\pi\frac{\Phi}{\phi_0} \right) - \frac{1}{2 L} (\Phi - \phi_{\mathrm{ext}})^2,
\end{equation}
where $\Phi = \phi_{v_2}-\phi_{v_1}$, $Q = q_{e_1}$,  and $\phi_{\mathrm{ext}}= - L M I_s$.  We have neglected an overall constant contribution to $L$. Thus, the Hamiltonian of the circuit reads 
\begin{equation} \label{eq:endof44}
    H = \frac{Q^2}{2C} - E_J \cos \left(2\pi\frac{\Phi}{\phi_0} \right) + \frac{1}{2 L} (\Phi - \phi_{\mathrm{ext}})^2,
\end{equation}
where the sole conjugate pair is consists of $Q$ and $\Phi$. Again, $\phi_{\mathrm{ext}}$ can be time-dependent.

\subsection{Time-dependent external charges or fluxes}

In the example above, we introduced the external flux by inductively coupling the circuit to an external current source. There is an alternative way in our description to introduce external flux: we can add one or more new branches with a voltage source to a loop. To understand this construction, we recall that Faraday's law states that in the presence of time-dependent magnetic fields, the sum of voltages in a loop equals the rate of change of the magnetic field. Thus, if $e_i$ ($i=1,2,\dots,n$) are the physical capacitive branches in a loop
\begin{equation}
	\sum_{i=1}^n \phi_{e_i} + \phi_{\mathrm{ext}} = 0,
\end{equation}
where $\phi_{\mathrm{ext}}$ is the external flux piercing the loop. This suggests that we can think of the external flux as just another branch in the loop with an additional fixed flux $\phi_{\mathrm{ext}}$. However, a natural question arises at this point: where should one put this additional branch in the loop? As discussed in Refs.~\cite{jens_timedep,PhysRevApplied.19.034031}, the various Hamiltonians are linked by a gauge transformation. 

\begin{figure}
    \centering
    \includegraphics[width = 8.6cm]{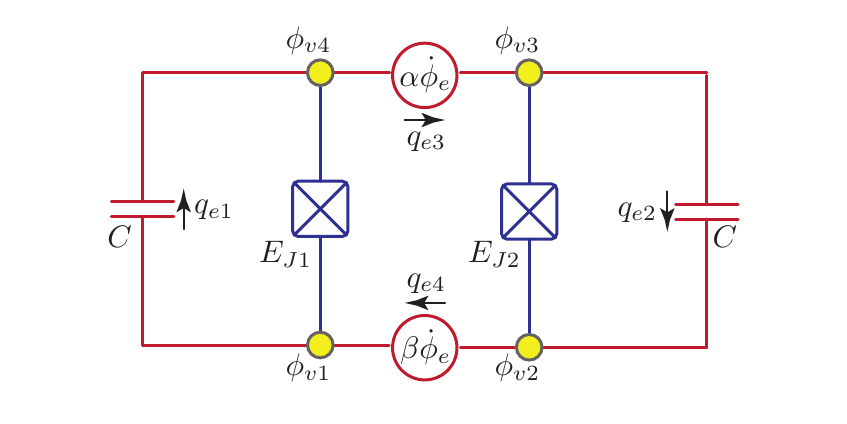}
    \caption{\textbf{Circuit in time-dependent external flux.} A flux-tunable transmon with two Josephson junctions shunted by two capacitors. The external flux in the loop enclosed by the two junctions can be modeled as one or more additional batteries in a loop, as long as the total flux provided by these batteries equals the external flux. The batteries are modeled as capacitive elements.} 
    \label{fig:phiext}
\end{figure}

Figure \ref{fig:phiext} shows an example of how one can place the ``flux batteries'' in a circuit to capture the external flux. The circuit is flux-tunable transmon, where two Josephson junctions, $E_{J1}$ and $E_{J2}$, are shunted by capacitors $C$. Notice that we have the freedom to place the external flux batteries in various ways, for example, here we choose to put two batteries with fluxes of $\alpha\phi_{\mathrm{ext}}$ and $\beta\phi_{\mathrm{ext}}$ in the loop. The condition of $\alpha+\beta=1$ ensures that the total external flux in the loop is $\phi_{\mathrm{ext}}$. This approach makes the circuit artificially a four-node circuit, but using the constraints outlined in this paper, we can end up with a single degree of freedom. 

To start, we recall that batteries are capacitive elements in our formalism, and using the variables in Fig.~\ref{fig:phiext}, we write down the Lagrangian
\begin{equation}
\begin{aligned}
	L &= q_{e_1}(\dot\phi_{v_4}-\dot\phi_{v_1}) + q_{e_2}(\dot\phi_{v_2}-\dot\phi_{v_3}) +q_{e_3}(\dot\phi_{v_3}-\dot\phi_{v_4}) + q_{e_4}(\dot\phi_{v_1}-\dot\phi_{v_2})+ \\
 & + E_{J1}\cos\left(2\pi\frac{\phi_{v_1}-\phi_{v_4}}{\phi_0} \right)+ E_{J2}\cos\left(2\pi\frac{\phi_{v_3}-\phi_{v_2}}{\phi_0} \right) - \frac{1}{2C}\left(q_{e_1}^2+q_{e_2}^2\right) + \\
 & - \beta\dot\phi_{\mathrm{ext}}q_{e_4} - \alpha\dot\phi_{\mathrm{ext}}q_{e_3},
\end{aligned}
\end{equation}
where the last two lines are the contribution of the voltage sources. Using the constraint arising from the capacitive loop $\Delta_C=\{\{e_1,e_4,e_2,e_3\}\}$, and integrating out $q_{e_3}$ and $q_{e_4}$, we arrive at the Lagrangian 
\begin{equation}
    L=Q\dot\Phi-\frac{1}{2}(\alpha-\beta)Q\dot\phi_{\mathrm{ext}}+ E_{J1}\cos\left(2\pi\frac{\Phi-\alpha\phi_{\mathrm{ext}}}{\phi_0} \right)+E_{J2}\cos\left(2\pi\frac{\Phi+\beta\phi_{\mathrm{ext}}}{\phi_0} \right)-\frac{1}{4C}Q^2,
\end{equation}
where $Q=q_{e_1}-q_{e_2}$ and $\Phi = \phi_{v_3}-\phi_{v_1}$. And finally, the Hamiltonian is 
\begin{equation}
    H=\frac{1}{4C}Q^2 - E_{J1}\cos\left(2\pi\frac{\Phi-\alpha\phi_{\mathrm{ext}}}{\phi_0} \right) - E_{J2}\cos\left(2\pi\frac{\Phi+\beta\phi_{\mathrm{ext}}}{\phi_0} \right)+\frac{1}{2}(\alpha-\beta)Q\dot\phi_{\mathrm{ext}},\label{eq:H_phi_ext}
\end{equation}
with conjugate pairs $\{\Phi,Q\}=1$.

Notice that the last term in Eq.~(\ref{eq:H_phi_ext}) depends on our choice of how to distribute the batteries in the loop. For example, in the ``irrotational gauge''~\cite{jens_timedep,PhysRevApplied.19.034031}, when $\alpha=\beta=\frac{1}{2}$, there is no term linear in $Q$. We can transform between the different gauges at the classical level by a time-dependent type-2 canonical transformation from $(Q,\Phi) \rightarrow (Q^\prime,\Phi^\prime)$. For example, if we take $\alpha=1$ and $\beta=0$, the Hamiltonian is 
\begin{equation}
    H=\frac{1}{4C}Q^2 - E_{J1}\cos\left(2\pi\frac{\Phi-\phi_{\mathrm{ext}}}{\phi_0} \right) - E_{J2}\cos\left(2\pi\frac{\Phi}{\phi_0} \right)+\frac{1}{2}Q\dot\phi_{\mathrm{ext}},
\end{equation}
while in the irrotational gauge
\begin{equation}
    \tilde H=\frac{1}{4C}\tilde Q^2 - E_{J1}\cos\left(2\pi\frac{\tilde\Phi-\frac{1}{2}\phi_{\mathrm{ext}}}{\phi_0} \right) - E_{J2}\cos\left(2\pi\frac{\tilde \Phi+\frac{1}{2}\phi_{\mathrm{ext}}}{\phi_0} \right).
\end{equation}
In this particular example, the generating function is \begin{equation}
    G(\Phi,\tilde Q,t) = \tilde Q \cdot \left(\Phi - \frac{\phi_{\mathrm{ext}}}{2} \right),
\end{equation}
which implies \begin{subequations}
    \begin{align}
        Q &= \frac{\partial G}{\partial \Phi} = \tilde Q, \\
        \tilde \Phi &= \frac{\partial G}{\partial \tilde Q} = \Phi-\frac{1}{2}\phi_{\mathrm{ext}}, \\
        \tilde H &= H + \frac{\partial G}{\partial t} = H - \frac{1}{2}Q\dot\phi_{\mathrm{ext}}.
    \end{align}
\end{subequations}

\subsection{Singular circuits}\label{sec:singular}

\begin{figure}
    \centering
    \includegraphics[width = 8.6cm]{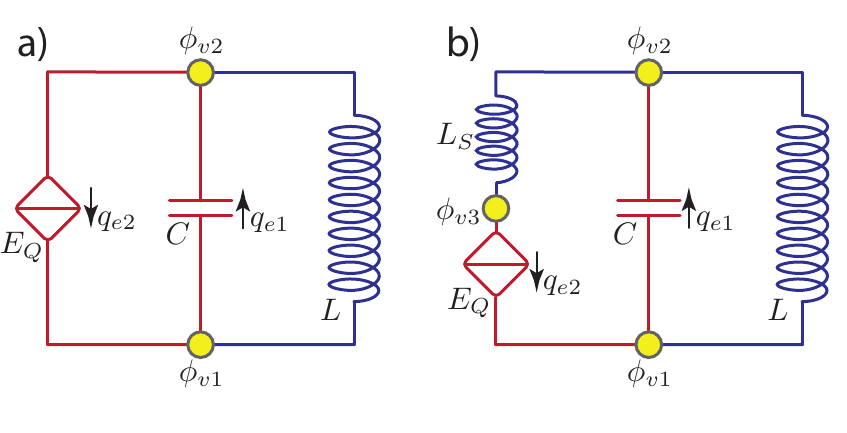}
    \caption{\textbf{Singular circuit.} (a) Example for a one-mode circuit where a quantum phase element has no series inductance attached to it. The circuit has a capacitive loop, $\mathcal{C}=\{\{q_{e_1},q_{e_2}\}\}$ leading to a non-analytical constraint. (b) When a series inductance is included in the circuit, the capacitive loop is broken, and the system has two degrees of freedom and an analytical Hamiltonian.} 
    \label{fig:singular}
\end{figure}

When a Josephson junction is not accompanied by a parallel capacitor, or a quantum phase slip element has no series inductor attached to it, the resultant circuit can become singular~\cite{rymarz_consistent_2022}. In this section, we analyze an example for such a singular circuit, which leads to a non-analytical Hamiltonian. 

The circuit is presented in Fig.~\ref{fig:singular}(a), and it has a quantum phase slip, a capacitor and an inductor all in parallel. Following our procedures, we arrive at a Lagrangian of 
\begin{equation}
  L = q_{e_1}(\dot\phi_{v_2}-\dot\phi_{v_1}) + q_{e_2}(\dot\phi_{v_1}-\dot\phi_{v_2}) - \frac{1}{2C}q_{e_1}^2 + E_Q\cos\left(2\pi\frac{q_{e_2}}{2e}\right) - \frac{1}{2L}\left(\phi_{v_1}-\phi_{v_2}\right)^2.
  \end{equation}
By looking at the geometry of the circuit, we notice that there is one capacitive loop $\mathcal{C}=\{\{q_{e_1},q_{e_2}\}\}$, which based on Eq.~(\ref{eq:constraint_left}) gives rise to the constraint
\begin{equation}
	V_Q\sin\left(2\pi\frac{q_{e_2}}{2e}\right) + \frac{q_{e_1}}{C} = 0.
	\label{eq:singular_circuit}
\end{equation}
This constraint, however, connects the two branch charges in a multi-valued, non-analytical way, leading to a Hamiltonian that is nonanalytical and can be evaluated only using numerical approaches. A detailed discussion on this topic can be found in Ref.~\cite{rymarz_consistent_2022}.
We denote 
\begin{equation}
    Q = q_{e_1} - q_{e_2}
\end{equation}
and write 
\begin{equation}\label{eq:singular_circuit_2}
    V_Q \sin \left(2 \pi\frac{q_{e_2}}{2 e} \right) + \frac{Q + q_{e_2}}{C} = 0 .
\end{equation}
By denoting the solution of Eq.~(\ref{eq:singular_circuit_2}) as $q_{e_2}\left(Q\right)$, and introducing $\Phi=\phi_{v_2}-\phi_{v_1}$, the Hamiltonian becomes
\begin{equation}
	H= \frac{1}{2C}\left[Q+q_{e_2}(Q)\right]^2  -E_Q\cos\left[2\pi\frac{q_{e_2}(Q)}{2e}\right]+ \frac{1}{2L}\Phi^2
\end{equation}
where the conjugate pairs are $\{\Phi,Q\}=1$.

In circuits which can be realized in current experiments, this singular behavior is not present because quantum phase slip elements always have a series inductor component $L_S$ [see Fig.~\ref{fig:singular}(b)].
This additional element transfers the singular one-mode circuit into a two-mode circuit that is analytical. The key observation is that the presence of the series inductance breaks the capacitive loop, and removes the left null vector of the circuit, thus, the constraint of Eq.~(\ref{eq:singular_circuit}) is lifted.  After a few steps, the Hamiltonian reads
\begin{equation}
    H = \frac{1}{2C}Q_{e_1}^2 + \frac{1}{2L}  \Phi_{e_1}^2 + \frac{1}{2 L_S} (\Phi_{e_1} + \Phi_{e_2})^2 - E_Q \cos\left[2 \pi \frac{Q_{e_2}}{2 e}\right]
\end{equation}
where 
\begin{equation}
    \begin{aligned}
        Q_{e_1} &= q_{e_1}, \\ 
        Q_{e_2} &= q_{e_2}, \\ 
        \Phi_{e_1} &= \phi_{v_2} - \phi_{v_1}, \\
        \Phi_{e_2} &= \phi_{v_1} - \phi_{v_3} ,
    \end{aligned}
\end{equation}
and the conjugate pairs are $\{\Phi_{e_1},Q_{e_1}\}=1$ and $\{\Phi_{e_2},Q_{e_2}\}=1$. A similar argument can be made for the case of parallel capacitors for Josephson junctions.

\subsection{Connecting to the existing Lagrangian formalism}
Finally, let us briefly discuss how our formalism straightforwardly reproduces the existing Lagrangian formalism in suitable limits.  As one example, consider a circuit with only linear capacitors, and inductive elements of any kind. The linear capacitor at the capacitive branch $e\in \mathcal C$ has capacitance $C_e$, while the inductive element at an inductive branch $e\in \mathcal{I}$ is described by an energy function $g_e$. In our formalism, the Lagrangian is
    \begin{equation}
    L(q_e,\phi_v,\dot\phi_v) = \sum_{e\in \mathcal C, v\in \mathcal V} q_e \Omega_{ev}\dot \phi_v - \sum_{e\in \mathcal C} \frac{q_e^2}{2C_e} - \sum_{e\in \mathcal{I}} g_e (\phi_e),
    \end{equation}
    where $\Omega_{ev}$ is the usual capacitive incidence matrix.
    Since $L$ is a quadratic function of $q_e$, but $L$ does not depend on the time derivatives of the charges $\dot{q}_e$, we can easily ``integrate out" $q_e$.  (This statement remains true in a quantum mechanical path integral).  Solving the Euler-Lagrange equation for $q_e$, we find \begin{equation}
        \Omega_{ev}\dot\phi_v = \frac{q_e}{C_e}.
    \end{equation}
    Then defining a capacitance matrix as  \begin{equation}
    C_{uv} = \sum_e C_e \Omega_{eu}\Omega_{ev},
    \end{equation}
    we find a Lagrangian expressed only in terms of node flux variables, as is standard in the literature \cite{vool_introduction_2017}:
    \begin{equation}
        L(\phi_v,\dot\phi_v) = \sum_{u,v\in \mathcal V} \frac{1}{2}C_{uv} \dot \phi_u \dot \phi_v - \sum_{e\in \mathcal{I}} g_e(\phi_e). \label{eq:capL}
    \end{equation}

\section{Outlook}
In this paper, we have developed a universal theory of circuit quantization for all LC circuits.  Our approach allows for the quantization of singular circuits, with arbitrary graph topology, and with arbitrary time-dependent sources.  The approach is inspired by symplectic geometry and graph theory, and the quantization prescription depends only on the topology of the capacitive subgraph, but not on which elements are linear or nonlinear. The ``spanning tree construction" leads to a straightforward quantization prescription using a set of canonically conjugate coordinates that could be efficiently implemented in future software packages that perform generic circuit quantization.

Looking forward, this approach will provide an efficient algorithm for computing the numerical spectra of complicated hybrid circuits simultaneously involving Josephson junctions, quantum phase slips, and other arbitrary nonlinear elements that can be classified as inductive or capacitive.  Obtaining these spectra will be a critical step in identifying the behavior of quantum phase slip and other nonlinear capacitive elements in a circuit quantum electrodynamics setup, opening new avenues to design and create novel superconducting devices beyond the current architectures. Further, this formalism should extend naturally to nonlinear mechanical oscillators \cite{Teufel2011,Khorasani2017, Ma2021,PhysRevX.11.031027,https://doi.org/10.48550/arxiv.2211.07632} 
and ideal non-reciprocal elements \cite{PhysRevB.99.014514}, and may also be extensible to the quantization of transmission lines \cite{ParraRodriguez2022canonical}.

\section*{Acknowledgements} 

This work was supported by a Research Fellowship from the Alfred P. Sloan Foundation under Grant FG-2020-13795 (AL), by the U.S. Air Force Office of Scientific Research under Grant FA9550-21-1-0195 (AO, AL) and by the U.S. Army Research Office under Grant No.~W911NF-22-1-0050 (SJ, TL, RWS, AG).

\begin{appendix}

\section{Mathematical formalism} \label{app:math}
In this appendix, we provide a mathematically precise discussion of the results highlighted in the main text.  First, we prove the necessary graph theoretic properties of $\Omega_{ev}$.  Second, we prove facts about circuit quantization, including the existence of the ``spanning tree construction" of canonical coordinates for arbitrary circuits, and identifying the number of Noether conserved charges.
\subsection{Graph theory for circuits}\label{app:gt}
In previous discussions, we have elected to use terminology familiar in the superconducting circuit community. We will continue to do so here. For the mathematically inclined, we refer to graph--theoretical edges as branches and we refer to vertices as nodes. 
  \begin{defn}
    Let $\mathcal V$ and $\mathcal E$ be discrete node and branch sets.  \textbf{Chains} on $\mathcal V$ are elements of
    \begin{equation}
        \mathcal{D}(\mathcal V) = \left\lbrace \left.\sum_{v\in \mathcal V} m_v |v\rangle \right| m_v \in \mathbb{Z} \right\rbrace
    \end{equation}
    where $|v\rangle$ is a vector in a vector space of dimension $|\mathcal V|$.  Chains on $\mathcal E$ are elements of a set $\mathcal{D}(\mathcal E)$ defined analogously. 
  \end{defn}
  \begin{defn}
     Define the \textbf{incidence map} $\mathcal{A}:E \rightarrow \mathcal V \times \mathcal V$, such that if $\mathcal{A}(e) = (u,v)$, branch $e$ is oriented from $u$ to $v$.  We'll really be using the \textbf{incidence matrix} $A : \mathcal{D}(\mathcal V) \rightarrow \mathcal{\mathcal D}(E)$, defined (using bra-ket notation) as 
      \begin{equation}
          \langle e|A|v\rangle = \left\lbrace\begin{array}{ll}
               1 &\ \mathcal{A}(e) = (u,v) \text{ for some } u   \\
               -1 &\ \mathcal{A}(e) = (v,u) \text{ for some } u   \\
               0 &\ \text{otherwise} 
          \end{array}\right..
      \end{equation}
      In the framework of homology, the incidence matrix $A$ can be thought of as the  boundary map $\partial$.  Let $A_{ev}= \langle e|A|v\rangle$.  Less formally, $A_{ev}=1$ if edge $e$ ends on $v$, while $A_{ev}=-1$ if $e$ starts on $v$.   We say that $(\mathcal V,\mathcal E,A)$ is a \textbf{directed graph}. 
  \end{defn}
  \begin{defn}
    Let $G = (\mathcal V,\mathcal E,A)$  be a directed and connected
    graph. 
    Partition $\mathcal E$ so that $\mathcal E = \mathcal C \cup \mathcal I$ and $ \mathcal C \cap \mathcal I = \emptyset$.
    We define $(\mathcal V,\mathcal E,A,\mathcal C)$ as a \textbf{circuit}.  Intuitively, we will put a capacitive element on all branches in $\mathcal C $, and an inductive element on all branches in $\mathcal I$.
  \end{defn}
  We note that the choice to require $G$ to be connected is without loss of generality since a circuit with multiple connected components would lead to a separable problem in later discussion.
  The object which we define to be a circuit consists of a directed graph and a determination of which branches contain capacitors. 
  \begin{defn}
      The \textbf{reduced incidence matrix} $\Omega : \mathcal{D}(\mathcal V) \rightarrow \mathcal{D}(\mathcal C)$ obeys $\langle e|\Omega|v\rangle = \langle e|A|v\rangle$.  The dimensions of $\Omega$ are distinct, and in what follows $\Omega_{ev} = \langle e|\Omega|v\rangle$ will appear frequently.
  \end{defn}
It is possible that multiple branches begin and end at the same place: i.e. we could have $\mathcal{A}(e_1) = \mathcal{A}(e_2)$ for some $e_1\ne e_2$.  We could also have $e_1 \in \mathcal  I$ and $e_2\in \mathcal C$.  

The following definitions will prove useful in what follows:
  \begin{defn}
      An \textbf{(unoriented) cycle} of length $l$, $\delta= (e_1,\ldots, e_l)$, is an ordered list of $l$ branches in $\mathcal C$, with the property that there exist nodes $v_1,\ldots, v_l$ such that $\mathcal{A}(e_1) = \lbrace v_1,v_2\rbrace$, $\mathcal{A}(e_2) = \lbrace v_2,v_3\rbrace $, $\ldots$, $\mathcal{A}(e_l) = \lbrace v_l, v_1\rbrace$. Here the notation means we do not care about the ordering: $\lbrace u,v\rbrace = \lbrace v,u\rbrace$. Indeed, we do not care about the orientation of edges on a cycle for this discussion.  Intuitively, $\delta$ is a loop made out of only capacitors in the circuit.  Let $\Delta_C$ denote the set of cycles in $\mathcal C$.   For the length $l$ cycle defined above, we write \begin{equation}
        |\delta\rangle = |(e_1,\ldots, e_l) \rangle = \sum_{j=1}^l \sigma_j |e_j\rangle \in \mathcal{D}(\mathcal C),
      \end{equation}
      where \begin{equation}
          \sigma_j = \left\lbrace\begin{array}{ll}
             1  &\ \mathcal{A}(e_j) = (v_j,v_{j+1}).  \\
             -1  &\ \mathcal{A}(e_j) = (v_{j+1},v_j).
          \end{array}\right..
      \end{equation}
      We remark that the choice of $\sigma_i$ implies that, for $\delta$ in $\Delta_C$,  $\langle \gamma | A | v\rangle = 0 $ for all $|v\rangle$ in $\mathcal D(\mathcal V)$. Qualitatively, one could say that $\sigma_i$ is chosen so that the chain corresponding to a given cycle is in some manner ``oriented" even if the cycle itself is not. 
  \end{defn}
  \begin{defn}
    An \textbf{inductively shunted island} in circuit $(\mathcal V,\mathcal E,A,\mathcal C)$ is a subset $U\subseteq \mathcal V$ with the following property:  $u_{1}\in U$ and $u_2\in U$ both hold if and only if there exists a path from $u_1$ to $u_2$ along branches in $\mathcal C$.  Note that a node only connected to inductors forms its own island.  Let $\Gamma_I(\mathcal V,\mathcal E,A,\mathcal{C})$ (denoted as $\Gamma_I$ for shorthand hereafter) be the set of all such islands.  For island $i \in \Gamma_I$, write \begin{equation}
          |i\rangle = \sum_{v \in i } |v\rangle \in \mathcal{D}(V).
      \end{equation}
  \end{defn}
  We make an analogous definition for inductors, which will be useful in future discussion. 
\begin{defn}
    Let $(\mathcal E,\mathcal V,\mathcal C,A)$ be a circuit.  Consider the undirected graph $G_L=(\mathcal V,\mathcal I)$ formed out of only the inductive branches.  We say that an \textbf{capacitvely shunted island} $j$ is a connected subgraph of $G_L$, which is not a proper subgraph of any other connected subgraph of $G_L$.  If the set of all such capacitively shunted islands is $\Gamma_C$, then \begin{equation}
        G_L = \bigcup_{j\in \Gamma_C} (j\cap \mathcal V, j\cap \mathcal I). 
    \end{equation}
\end{defn}
We remark that $G_L$ contains all nodes that are in $G$ and in particular even those connected to no inductors. Such nodes constitute their own capacitively shunted island in the same manner that nodes connected only to capacitors make up their own capacative island. 
The set of inductive islands derived from some circuit is unique. 

Combining all of the definitions above, the following theorem summarizes the critical properties of $\Omega$.
\begin{thm}
    The generically non-square matrix $\Omega$ has the following properties: \begin{enumerate}
        \item $\langle \alpha|\Omega = 0$ (i.e. $\langle \alpha |$ is a left null vector of $\Omega$) if and only if \begin{equation}
            |\alpha\rangle = \sum_{\delta \in \Delta_C} \alpha_\delta|\delta\rangle, \;\;\; \alpha_\delta\in \mathbb{Z}.  \label{eq:alpha}
        \end{equation}
        Note that the right hand side in general contains linearly dependent vectors.
        \item $\Omega|\beta\rangle = 0$ (i.e. $|\beta\rangle$ is a right null vector of $\Omega$) if and only if \begin{equation}
            |\beta\rangle = \sum_{i\in \Gamma_I} \beta_i |i\rangle, \;\;\; \beta_i \in \mathbb{Z}.\label{eq:beta}
        \end{equation}
        \item $\Omega$ contains a non-degenerate submatrix $\bar\Omega$, which is a $n\times n$ square matrix where \begin{equation}
            n = |\mathcal V| - |\Gamma_I|.
        \end{equation}
    \end{enumerate}\label{thm110}
\end{thm}
\begin{proof}
    We prove the three parts in turn: \begin{enumerate}
        \item It is straightforward to see that (\ref{eq:alpha}) is a null vector of $\Omega$, using linearity and the fact that for any cycle $\delta$ of length $l$, \begin{equation}
        \langle \delta|\Omega = \left(\langle v_1| - \langle v_2| \right) + \left(\langle v_2| - \langle v_3| \right) + \cdots + \left(\langle v_l| - \langle v_1| \right) = 0.
        \end{equation}
        The converse is implied by the homology of the undirected graph $(\mathcal V,\mathcal C)$, but we prove it explicitly.  It will prove convenient to (without loss of generality) choose the orientations of edges $|e\rangle$ such that for all $e$, $\langle \alpha|e\rangle \ge 0$.\footnote{This can be done by simply defining $|-e\rangle = -|e\rangle$ as an edge oriented in the opposite direction.} Pick an edge $e$ such that $\langle \alpha |e\rangle > 0 $. If $\langle e|\Omega = \langle v| - \langle u|$, then there must exist an edge $e^\prime$, with $\langle \alpha|e^\prime \rangle \ne 0$ and $\langle e^\prime |\Omega|v\rangle = -1$, since $\langle e|\Omega |v\rangle =1$ but $\langle \alpha | \Omega |v\rangle = 0$.  Build the bra $\langle \psi| = \langle e| + \langle e^\prime|$, and observe that \begin{equation}
            \langle \psi|\Omega = \langle v^\prime | - \langle u|.
        \end{equation}
        If $\langle v^\prime| \ne \langle u|$, then we keep going: look for another edge in $\langle \alpha|$ of the form $\langle e^{\prime\prime}|\Omega = \langle v^{\prime\prime}| - \langle v^\prime|$ -- such an edge must exist since $\langle \alpha|$ is a null vector, etc.  Then update $|\psi\rangle \rightarrow |\psi\rangle + |e^{\prime\prime}\rangle$.   Since in each step of this process, \begin{equation}
            \sum_e \langle e|\psi\rangle \rightarrow 1+ \sum_e \langle e|\psi\rangle ,
        \end{equation} eventually, this process must terminate because we will (if it does not) run out of edges in $\langle \alpha|$ to include.   When the process does terminate and we find a left null vector $\langle \psi|$ with $\langle \psi|\Omega=0$.  If $\langle \alpha| - \langle \psi| = 0$, then $\langle \alpha|$ simply corresponds to a chain associated with a cycle $\delta \in \Delta_C$.   Otherwise, $\langle \alpha| - \langle \psi|$ is a non-trivial vector where we can simply repeat the argument.   Since the sum of coefficients of $\langle \alpha| - \langle \psi|$ is smaller than the sum in $\langle \alpha |$, the process will terminate.  By construction, each $\langle \psi|$ that we find in this process formed a cycle $\delta \in \Delta_C$, meaning our null vector can be expressed as a chain of the form (\ref{eq:alpha}).
        
        \item Pick some $e\in \mathcal C$.  We know that $\langle e|\Omega|\beta\rangle=0$, which means that if $\mathcal{A}(e) = (u,v)$, then $\langle u|\beta\rangle \beta_u=\beta_v = \langle v|\beta$.  For any two vertices in island $i\in \Gamma_I$, we can find a path between them (call its length $l$): $(e_1,\ldots, e_l)$.  Applying this argument to all edges along the path, we conclude that if $v_{1,2}\in i$, $\beta_{v_1}=\beta_{v_2}$.  Thus, the most generic null vector is of the form (\ref{eq:beta}).   It is also straightforward to check that (\ref{eq:beta}) is always a right null vector.

        \item This follows from the rank-nullity theorem, the fact that $\mathcal{D}(\mathcal V)$ is a $|\mathcal V|$-dimensional vector space, and the fact that the number of linearly independent null vectors in (\ref{eq:beta}) is $|\Gamma_I|$.
    \end{enumerate}
    Thus we prove all three claims.
\end{proof}
The substance of Theorem \ref{thm110} is that it is possible to define a symplectic form (and thus to quantize) immediately after enumerating the null vectors of $\Omega$, which are counted exhaustively by the theorem above. Later results will provide simplifications to circuits in general, but none of the following results are strictly necessary to generically quantize non--dissipative circuits. 
Point 3 of the theorem above immediately implies the following corollaries:
\begin{cor}
    Consider the circuit $(\mathcal V,\mathcal E,A,\mathcal C)$. Define 
    \begin{equation}
       g = |\mathcal C| - |\mathcal V| + 1.
    \end{equation}
    $g$ is the graph--theoretic equivalent of a topological genus, and it counts the number of loops in a graph. 
    \begin{equation}
        g - 1 = |\Delta_C| - |\Gamma_I| .
    \end{equation}
\end{cor}
\begin{proof}
    The rank nullity theorem guarantees that the row--rank and the column--rank of $\Omega$ are equal, and $|\mathcal V| - |\Gamma_I| $ counts the row--rank of $ \Omega$ while $|\mathcal E| - |\Delta_C| $ counts the column--rank therein. 
\end{proof}

\subsection{Circuit quantization on general graphs}\label{app:sym}
The discussion in Appendix \ref{app:gt} was entirely self--contained. At this point, we shift our focus to the relation of graph theory to the dynamical systems of interest. What follows will depend on discussion that can be found broadly in 
Section \ref{sec:formal}.

\begin{thm}\label{thm330}
    Let $G = (\mathcal E,\mathcal V,A,\mathcal C)$ be a circuit. It is always possible to define $|\mathcal C| - |\Delta_C|$ variables
    $Q_i = A_{ie}q^e$ and $\Phi_i = B_{iv}\phi^v$ so that 
    \begin{equation}
        \sum_{e,v}q_e \Omega_{ev} \dot \phi_v = \sum_{i = 1}^{|\mathcal C| - |\Delta_C|} Q_i\dot\Phi_i =\sum_{i = 1}^{|\mathcal C| - |\Delta_C|} \sum_{e,v} A_{ie} B_{iv} q_e \dot\phi_v.
    \end{equation}
  $A_{ie}$ and $B_{iv}$ can be obtained by picking a spanning tree $\mathcal{T}\subseteq \mathcal{C}$.   The transformations between the variables corresponding to different spanning trees are canonical.
\end{thm}
\begin{proof}
    Choose a spanning tree $\mathcal T$ of $\mathcal C$ and write 
    \begin{equation}
        \mathcal C \setminus \mathcal T = \{a_1, a_2 ,\dots , a_{|\Delta_C|}\}.
    \end{equation}
    The spanning tree has the property that every node adjacent to some branch in $\mathcal C$ is adjacent to some branch in $\mathcal T$, and $\mathcal T$ contains no cycles. 
    Succinctly, if $v_1$ and $v_2$ are connected by capacitors, there is a unique shortest path from $v_1$ to $v_2$ traversing only on branches in $\mathcal T$. The choice of $\mathcal T$ is necessarily non--unique and we will later show that the choice of $\mathcal T$ cannot  have physical implications. 
    By the uniqueness in $\mathcal T$ of a path between two vertices, every branch flux from $\mathcal C \setminus \mathcal T$ can be expressed as an linear combination of branch fluxes from $\mathcal T$ with integer coefficients. 
 (Recall the definition of branch flux $\phi_e$ in (\ref{eqn:branchflux}).) Namely, there exists some (generally rectangular) matrix $K\in \mathbb R ^{|\mathcal T | \times |\mathcal C \setminus \mathcal T|}$ satisfying 
    \begin{equation}
       \phi_e = \sum_{f \in \mathcal T} K_{ef}  \phi_f \label{eq:Kmatrix}
    \end{equation}
    for $e\in\mathcal C \setminus \mathcal T$. By construction, $K_{ef}\in \lbrace -1,0,1\rbrace$.

    We recognize that 
    \begin{equation}
        \sum_{e,v} q_e \Omega_{ev} \dot\phi_v = \sum_{f \in \mathcal T} q_f \dot \phi_f + \sum_{e \not\in \mathcal T}\sum_{f \in \mathcal T} q_e K_{ef} \dot\phi_f.
    \end{equation}
    Define 
    \begin{equation}
        M = 
        \begin{pmatrix}
            \mathbb I_{|\mathcal T|} & K
        \end{pmatrix}^\mathrm{T}
    \end{equation}
    and rewrite 
    \begin{equation}
        Q_f = \sum_{e}q_e M_{ef}.
    \end{equation}
    Since $M$ has rank $|\mathcal T|$, the $Q_f$ variables are independent, and they span a vector space of dimension $|\mathcal T| = |\mathcal C| - |\Delta_C|$, so we see that all left null vectors of $\Omega_{ev}$ have been removed by this definiton. 
    A similar argument shows that the $\phi_f$ variables are also independent, and thus all right null vectors have been removed. 

    Now suppose another spanning tree, $\mathcal T'$, is chosen \footnote{There are $M \le \prod_{\delta \in \Delta_C}|\delta|$ possible choices of spanning tree in a circuit.}. Suppose further that $\mathcal T$ and $\mathcal T'$ differ by a single edge, so that 
    \begin{equation}
        \mathcal T = \{f_0\} \cup \mathcal A
    \end{equation} 
    and 
    \begin{equation}
        \mathcal T' = \{f_1\} \cup \mathcal A
    \end{equation}
    with both $f_0$ and $f_1$ absent from $\mathcal A \subset \mathcal C$.
    Further demand that $f_0$ and $f_1$ are in the same cycle. 
    Since $\mathcal T$ and $\mathcal T'$ are both spanning trees, there exists some matrix $\Sigma$ whose entries are elements of $\{1,-1,0\}$ so that 
    \begin{equation}
        \phi_{f} = \sum_{f'} \Sigma_{ff'} \phi_{f'}
    \end{equation}
    for $f'$  in $\mathcal T'$ and $f$ in $\mathcal T$. Moreover, necessarily $|\mathcal T | = |\mathcal T'|$ so $\Sigma$ is square. Moreover, $\Sigma$ is the identity on edges common to $\mathcal T$ and $\mathcal T'$.  
    Since $f_1$ is an element of $\mathcal T' \setminus \mathcal T$, 
    \begin{equation}
        \phi_{f_1} = \sum_{f \in \mathcal T} K_{f_1f}(\mathcal{T})\phi_f,
    \end{equation}
    where $K(T)$ denotes the $K$-matrix constructed in (\ref{eq:Kmatrix}) for $\mathcal{T}$.
    Since $f_0$ is the unique element of $\mathcal T\setminus \mathcal T'$ and $f_0$ is in the same cycle as $f'$, necessarily the element $K_{f_1f_0}(\mathcal{T})\ne 0$, as otherwise there would be a cycle in $\mathcal{T}$.  So
    \begin{equation}
        \phi_{f_1} = \sum_{f \neq f_0 \in \mathcal T} K_{f'f}(\mathcal{T}) \phi_f + K_{f_1 f_0}(\mathcal{T}) \phi_{f_0}
    \end{equation}
    and thus 
    \begin{equation}
        K_{f_1 f_0}(\mathcal{T}) \phi_{f_0} = \phi_{f_1} - \sum_{f \neq f_0 \in \mathcal T} K_{f_1f}(\mathcal{T})  \phi_f.
    \end{equation}
    Evidently, 
    \begin{equation}
        \Sigma_{ff'} = 
        \begin{pmatrix}
        \mathbb I _{|\mathcal T \cap \mathcal T'|} & 0 \\ 
        0 & K_{f_1 f_0}(\mathcal{T})
        \end{pmatrix}
        \begin{pmatrix}
        \mathbb I_{|\mathcal T \cap \mathcal T'|} & 0 \\
        K_{f_1 f}(\mathcal{T}) & 1 \\
        \end{pmatrix}.
    \end{equation}
    Schematically, $\Sigma = \mathbf{D}(\mathbb I + \mathbf{N})$ with $\mathbf{N}^2 = 0$ and $\mathbf{D}^2 = \mathbb I$. 
    Immediately, $\Sigma^{-1} = (\mathbb I - \mathbf{N}) \mathbf{D}$. 

    Now, clearly, 
    \begin{equation}
        \sum_{f \in \mathcal T} Q_f \dot\phi_f = \sum_{f \in \mathcal T}\sum_{f' \in \mathcal T'} Q_f \Sigma_{ff'} \phi_{f'}  = Q_{f'}' \dot\phi_{f'}
    \end{equation}
    with 
    \begin{equation}
        Q_{f'}' = \sum_{f\in\mathcal{T}} Q_f\Sigma_{f'}^f.
    \end{equation}
    Now the transformation 
    \begin{equation}
    \begin{aligned}
        \mathbf Q \rightarrow \mathbf Q'  = \mathbf Q  \Sigma \\
        \boldsymbol{\phi} \rightarrow \boldsymbol{\phi}' =  \Sigma^{-1} \boldsymbol{\phi}
        \end{aligned}
    \end{equation}
    is clearly canonical.

    Lastly, for general $\mathcal{T}$ and $\mathcal{T}^\prime$, we can find a sequence $\mathcal{T} \rightarrow \mathcal{T}_1\rightarrow \cdots \rightarrow \mathcal{T}^\prime$ where at each step, $\mathcal{T}_j$ and $\mathcal{T}_{j+1}$ differ by a single edge. The composition of canonical transformations corresponding to each step is still canonical.  
    \end{proof}
In the proof of Theorem \ref{thm330}, the $Q_f$ variables are independent and always number $|\mathcal C|  - |\Delta_C|$. Moreover, the transformation between different choices of spanning tree is canonical. Together, these two facts imply that the same non-dynamical degrees of freedom are removed by \emph{every} choice of spanning tree. By theorem \ref{thm110} these are necessarily those corresponding to the current about each loop of capacitors.  Similarly, the null vectors associated to inductively shunted islands are always the same.

Theorem \ref{thm330} depends only upon Theorem \ref{thm110} and previous graph--theoretic constructions. 
In practice, it is often the case that one naturally identifies conjugate variables by applying Theorem \ref{thm220} instead, without the need to apply the above result. 
Further, it may pose a technical challenge to write the Hamiltonian associated with some circuit in terms of only $Q_e$ variables as defined above (see Section \ref{sec:singular}), but as a matter of principle the Hamiltonian will always exist.

One utility of Theorem \ref{thm330} arises in the circumstance where one wishes to periodically identify some subset of variables on a finite interval. Upon doing so, it is required that the conjugate of the periodic variable be integer--valued and theorem \ref{thm330} guarantees that this is the case.  
We state the following two observations to emphasize this point: 

\begin{obs}
    \label{cor:sym}
    Consider the phase space 
    \begin{equation}
    M = \mathbb{R}^{|\mathcal V|}\times \mathbb{R}^{|\mathcal C|}.
    \end{equation}
    Given the reduced incidence matrix $\Omega$, define the 2-form on $M$ \begin{equation}
        \omega = \sum_{e\in \mathcal C}\sum_{v\in \mathcal V} \Omega_{ev} \mathrm{d}q_e \wedge \mathrm{d}\phi_v .
    \end{equation}
    Then there exists a submanifold $\bar M$ of $M$, which is symplectic, diffeomorphic to (equivalent to) the cotangent bundle $\mathbb{T}^* \mathbb R^n = \mathbb R^{2 n}$, with $\omega$ a globally constant symplectic form. \label{cor111}
    On $\bar M$, 
    \begin{equation}
        \omega = \sum_{f \in \mathcal T} \mathrm d Q_f \wedge \mathrm d \phi_f 
    \end{equation}
    as guaranteed by Theorem \ref{thm330}.  
\end{obs}

\begin{obs}\label{obs:quotient}
    Let $G = (\mathcal E,\mathcal V,A,\mathcal C)$ be a circuit, and let $\omega$ be the symplectic form $\omega$ on $T^*\mathbb R^n$ provided by Cor.~\ref{cor:sym} through Thm.~\ref{thm110}. If some number of variables are identified as periodic, the symplectic form on the resulting quotient manifold, $\omega'$ is given by the quotient map.
\end{obs}
    This result follows immediately from the quotient manifold theorem \cite{manifold}. We remark that the symplectic form produced by the quotient map is not exact, in the physically relevant case where we quotient by a free group action of $\mathbb{Z}^k$, so that $\mathbb{R}^k/\mathbb{Z}^k = \mathrm{T}^k$ becomes a torus.

Now, we count the number of degrees that can be removed by Noether's Theorem for a general circuit.

\begin{thm}\label{thm220}
  Let $(\mathcal E,\mathcal V,A,\mathcal C)$ be a circuit. Consider the symplectic form and Hamiltonian provided by (\ref{eq:main}).  If $\Gamma_C$ is the set of capacitively shunted islands, there exist $|\Gamma_C|-1$ Noether charges. 
\end{thm}
\begin{proof}
  Enumerate the (unique) elements of $\Gamma_C = \{J_1,J_2,\dots,J_{|\Gamma_C|}\}$. 
  Without loss of generality, take $J_1 = \lbrace u_1,\ldots, u_l\rbrace$. If any inductive edge couples to a $u_i\in J_1$, the edge also connects to another vertex in $J_1$ (by definition of a capacitively shunted island).  Therefore, since $H$ only depends on $\phi_v$s via inductive edges, \begin{equation}
      H(q_e,\phi_v) = H(q_e, \phi_v + c_v)
  \end{equation}
  where \begin{equation}
      c_v = \left\lbrace\begin{array}{ll} 1 &\ v\in J_1 \\ 0 &\ v\notin J_1\end{array}\right..
  \end{equation}
  From Lagrangian (\ref{eq:main}), evaluate the Euler-Lagrange equation \begin{equation}
      0 = \sum_{v\in\mathcal{V}} c_v \frac{\delta S}{\delta \phi_v} = - \sum_{v\in \mathcal{V}} c_v \left[ \sum_{e\in\mathcal{C}}\dot q_e \Omega_{ev} - \frac{\partial H}{\partial \phi_v} \right] = - \frac{\mathrm{d}}{\mathrm{d}t}\sum_{v\in J_1}\sum_{e\in\mathcal{C}}q_e\Omega_{ev}.
  \end{equation}
We conclude that \begin{equation}
    Q_{J_i} = \sum_{v\in J_i}\sum_{e\in \mathcal C} q_e\Omega_{ev}
  \end{equation}
  is a constant of motion.
  
  Since we can do this for all $|\Gamma_C|$ islands, we may naively conclude that there are $|\Gamma_C|$ independent Noether charges.  However, notice that \begin{equation}
  \sum_{J_i \in \Gamma_C} Q_{J_i} = \sum_{v\in\mathcal{V}}\sum_{e\in\mathcal{C}} q_e \Omega_{ev} = 0,   \label{eq:thm13vanishing}
  \end{equation}
  since every edge enters one vertex and exits one vertex.  Therefore, there is a constraint that not all $Q_{J_i}$s can be independent.  Since we assume that the circuit is connected, there will be no other constraints on any $Q_{J_i}$, since any subset of the capacitively shunted islands (that is not the entire circuit) must be connected to another part of the circuit by at least one capacitive edge, meaning that the generalization of (\ref{eq:thm13vanishing}) does not vanish for any other subset of $\Gamma_C$.
  \end{proof}

\section{Quantization without Noether charges}\label{app:noether}
We have remarked that Theorem \ref{thm110} is necessary to carry out the quantization procedure but that Theorem \ref{thm220} is not. We will provide an example that makes this distinction clear. 
The Lagrangian corresponding to the circuit drawn in Figure \ref{fig:examples_1}(b) is given by 
\begin{equation}
    L =  q_{e_1}(\dot \phi_{v_2} - \dot \phi _{v_1}) + q_{e_2} (\dot \phi_{v_3} - \dot \phi_{v_2} ) - \frac{1}{2 C_1} q_{e_1}^2 - \frac{1}{2 C_2} q_{e_2}^2 - \frac{1}{2 L} (\phi_{v_3} -\phi_{v_1} )^2. 
\end{equation}
To be explicit, the matrix $\Omega$ is given by 
\begin{equation}
    \Omega = \begin{pmatrix} -1 &  1 & 0 \\
    0 & 1 & -1 
    \end{pmatrix}
\end{equation}
and we see readily that $\Omega$ has the right null vector corresponding to a uniform sum over nodes, which informs us that the time evolution of $\phi_{v_1} + \phi_{v_2} + \phi_{v_3}$ is not fixed by $L$.  In the spanning tree construction, we choose dynamical variables $\Phi_{e_1}$ and $\Phi_{e_2}$ to be branch fluxes on the unique spanning tree; the standard branch charges $Q_{e_1}=q_{e_1}$ and $Q_{e_2}=q_{e_2}$ are the canonical conjugate variables. 
The Lagrangian becomes
\begin{equation}
    L = Q_{e_1}\dot\Phi_{e_1} + Q_{e_2}\dot\Phi_{e_2}
    - \frac{1}{2 C_1} Q_{e_1}^2 -  \frac{1}{2 C_2} Q_{e_2}^2 - \frac{1}{2 L } (\Phi_{e_1} + \Phi_{e_2})^2 .
\end{equation}
The quantum Hamiltonian is
\begin{equation}\label{eq:redundantham}
    H = - \frac{\hbar^2}{2 C_1}  \frac{\partial^2 }{\partial \Phi_{e_1}^2 } - \frac{\hbar^2}{2 C_2 } \frac{\partial^2 }{\partial \Phi_{e_2}^2 } + \frac{1}{2 L} (\Phi_{e_1} + \Phi_{e_2})^2 .
\end{equation}
We  recognize (\ref{eq:redundantham}) as a harmonic oscillator coupled to a free particle, after a suitable coordinate charge. The free particle degree of freedom could have been removed from the start by using Theorem \ref{thm220}: the Noether charge is associated with charge conservation on the subcircuit trapped between the two capacitors. 

\end{appendix}

\bibliography{thebib}

\end{document}